\documentclass[11pt, sort]{article}

\usepackage{hyperref, bm}
\usepackage{tabularx}
\usepackage{comment}
\usepackage{soul}
\usepackage{geometry}

 \usepackage{mathtools}

\makeatletter
\def\thm@space@setup{\thm@preskip=3pt
\thm@postskip=3pt}
\makeatother
\usepackage{titlesec}

\usepackage{graphicx}
\usepackage{amsmath, amsthm, amssymb, amsfonts, latexsym}
\usepackage{array}
\usepackage{caption}
\usepackage{color, colortbl}
\newcommand{\heston}{{\scalebox{0.5}{Heston}}}
\newcommand{\HQH}{{\scalebox{0.5}{HQH}}}

\usepackage{subfigure}
\usepackage{pgfplots}

\usepackage{epsfig}
\usepackage{fancybox}
\usepackage{array}
\usepackage{shadow}

\usepackage{mathrsfs}
\usepackage{longtable, relsize}
\usepackage{calc}
\usepackage{multirow}
\usepackage{hhline}
\usepackage{ifthen}
\usepackage{algorithmic}
\usepackage{algorithm}
\usepackage{color}

\newcommand{\siscblue}{\color{black}}
\newcommand{\siscb}{\color{blue}}

\usepackage[running, mathlines]{lineno}

\newcommand*\patchAmsMathEnvironmentForLineno[1]{%
  \expandafter\let\csname old#1\expandafter\endcsname\csname #1\endcsname
  \expandafter\let\csname oldend#1\expandafter\endcsname\csname end#1\endcsname
  \renewenvironment{#1}%
     {\linenomath\csname old#1\endcsname}%
     {\csname oldend#1\endcsname\endlinenomath}}%
\newcommand*\patchBothAmsMathEnvironmentsForLineno[1]{%
  \patchAmsMathEnvironmentForLineno{#1}%
  \patchAmsMathEnvironmentForLineno{#1*}}%
\AtBeginDocument{%
\patchBothAmsMathEnvironmentsForLineno{equation}%
\patchBothAmsMathEnvironmentsForLineno{align}%
\patchBothAmsMathEnvironmentsForLineno{flalign}%
\patchBothAmsMathEnvironmentsForLineno{alignat}%
\patchBothAmsMathEnvironmentsForLineno{gather}%
\patchBothAmsMathEnvironmentsForLineno{multline}%
}

\numberwithin{equation}{section}
\numberwithin{table}{section}
\numberwithin{figure}{section}

\newtheorem{definition}{Definition}
\newtheorem{theorem}{Theorem}
\newtheorem{lemma}{Lemma}
\newtheorem{remark}{Remark}

\newtheorem{corollary}{Corollary}

\numberwithin{definition}{section}
\numberwithin{theorem}{section}
\numberwithin{lemma}{section}
\numberwithin{remark}{section}
\numberwithin{assumption}{section}
\numberwithin{condition}{section}
\numberwithin{property}{section}
\numberwithin{proposition}{section}
\numberwithin{corollary}{section}
\numberwithin{algorithm}{section}
\usepackage{empheq}

\newcommand{\myblue}{\color{black}}

\newcommand{\EQ}{\begin{equation}}
\newcommand{\EN}{\end{equation}}
\newcommand{\EQS}{\begin{equation*}}
\newcommand{\ENS}{\end{equation*}}
\newcommand{\EQA}{\begin{eqnarray}}
\newcommand{\ENA}{\end{eqnarray}}
\newcommand{\EQAS}{\begin{eqnarray*}}
\newcommand{\ENAS}{\end{eqnarray*}}
\newcommand{\AL}{\begin{align}}
\newcommand{\AN}{\end{align}}
\newcommand{\ALS}{\begin{align*}}
\newcommand{\ANS}{\end{align*}}

\newcommand{\md}{d}

\newcommand{\Ibb}{\mathbb{I}}

\newcommand{\ds}{\displaystyle}
\def\r{\right}
\def\l{\left}
\def\f{\frac}

\newcommand{\Mcal}{\mathcal{M}}

\newcommand{\Pcal}{\mathcal{P}}

\newcommand{\Scal}{\mathcal{S}}

\newcommand{\Rbb}{\mathbb{R}}
\usepackage{wrapfig}
\usepackage{bm}
\usepackage{amsfonts}

\newcommand{\loc}{\scalebox{0.7}{$loc$}}

\newcommand{\norm}{\scalebox{0.5}{norm}}

\usepackage{xspace}
\usepackage{bold-extra}
\usepackage[most]{tcolorbox}

\colorlet{texcscolor}{blue!50!black}
\colorlet{texemcolor}{red!70!black}
\colorlet{texpreamble}{red!70!black}
\colorlet{codebackground}{black!25!white!25}

\def\r{\right}
\def\l{\left}

\newcommand{\eps}{\epsilon}
\newcommand{\G}{G(\cdot)}
\makeatletter
\newcommand{\subalign}[1]{%
  \vcenter{%
    \Let@ \restore@math@cr \default@tag
    \baselineskip\fontdimen10 \scriptfont\tw@
    \advance\baselineskip\fontdimen12 \scriptfont\tw@
    \lineskip\thr@@\fontdimen8 \scriptfont\thr@@
    \lineskiplimit\lineskip
    \ialign{\hfil$\m@th\scriptstyle##$&$\m@th\scriptstyle{}##$\crcr
      #1\crcr
    }%
  }
}
\makeatother

\AtBeginDocument{
\addtolength{\abovedisplayskip}{-0.2ex}
\addtolength{\abovedisplayshortskip}{-0.2ex}
\addtolength{\belowdisplayskip}{-0.2ex}
\addtolength{\belowdisplayshortskip}{-0.2ex}
}

\usepackage{titlesec}
\titlespacing*{\section}
  {0pt}{0.5\baselineskip}{0.2\baselineskip}
\titlespacing*{\subsection}
  {0pt}{0.5\baselineskip}{0.2\baselineskip}
\titlespacing*{\subsubsection}
  {0pt}{0.5\baselineskip}{0.2\baselineskip}

\begin{document}

\title{Fourier Neural Network Approximation of Transition Densities in~Finance}
\author{
Rong Du \thanks{School of Mathematics and Physics, The University of Queensland, St Lucia, Brisbane 4072, Australia,
email: \texttt{rong.du1@uq.net.au}
}
\and
Duy-Minh Dang\thanks{School of Mathematics and Physics, The University of Queensland, St Lucia, Brisbane 4072, Australia,
email: \texttt{duyminh.dang@uq.edu.au}
}
}
\date{\today}
\maketitle

\begin{abstract}
This paper introduces FourNet, a novel single-layer feed-forward neural network (FFNN) method designed to approximate transition densities for which closed-form expressions of their Fourier transforms, i.e.\ characteristic functions,  are available. A unique feature of FourNet lies in its use of a Gaussian activation function, enabling exact Fourier and inverse Fourier transformations and drawing analogies with the Gaussian mixture model.
We mathematically establish FourNet's capacity to approximate transition densities in the $L_2$-sense arbitrarily well with finite number of neurons. The parameters of FourNet are learned by minimizing a loss function derived from the known characteristic function and the Fourier transform of the FFNN, complemented by a strategic sampling approach to enhance training.
We derive practical bounds for the $L_2$ estimation error and the potential pointwise loss of nonnegativity in FourNet
{\siscb{for $d$-dimensions, $d \ge 1$,}} highlighting its robustness and applicability in practical settings. FourNet's accuracy and versatility are demonstrated through a wide range of dynamics common in quantitative finance, including L\'{e}vy processes and the Heston stochastic volatility models-including those augmented with the self-exciting Queue-Hawkes jump process.

\vspace*{.1in}

\noindent
{\bf{Keywords:}} transition densities, neural networks, Gaussian activation functions,  Fourier transforms, characteristic functions
\vspace{.1in}

\noindent\noindent {\bf{MSC codes:}} 62M45, 91-08, 60E10, 62P05
\end{abstract}
\section{Introduction}
\label{sec:intro}
The application of machine learning, especially deep learning, in quantitative finance has garnered considerable interest. Recent breakthroughs in computational resources, data availability, and algorithmic enhancements have encouraged the adoption of machine learning techniques in various quantitative finance domains. These include, but are not limited to, portfolio optimization \cite{li2019data, van2023parsimonious}, asset pricing \cite{yan2018financial, borovykh2017conditional}, model calibration and option pricing \cite{su2021option, liu2019neural, hinds2022neural}, solution of high-dimensional partial differential equations \cite{han2018solving, Pham2019, weinan2021algorithms, sirignano2018dgm}, valuation adjustments \cite{frode2022neural, gnoatto2023deep, goudenege2022computing}, as well as aspects of stochastic control and arbitrage-free analysis \cite{ito2021neural, reisinger2020rectified, cohen2021arbitrage}.

Transition (probability) density functions, which are crucial in quantitative finance due to their primary role in governing the dynamics of stochastic processes, often do not admit a closed-form expression. Consequently, the utilization of numerical methods for estimating these density functions becomes necessary. Classical methods include kernel density estimation, as referenced in \cite{milstein2004transition, rosenblatt1956remarks, giles2015multilevel}.
Yet, surprisingly, the development of neural network (NN) methods for estimating these transition probability density functions is significantly underdeveloped. While some existing NN strategies tackle the associated high-dimensional Kolmogorov partial differential equations (PDEs) using deep NNs, these are primarily black-box in nature. Such methodologies have seen applications in option pricing (\cite{su2021option, fddensity}) and general It\^{o} diffusions (\cite{gu2023stationary}).
While these methods are generally effective and versatile, they come with a major limitation: their model-dependent nature necessitates a constant reformulation of the Kolmogorov PDEs for different stochastic models. In addition, the inherent complexity associated with deploying NNs to solve PDEs might deter their practical application.
Furthermore, a notable gap in the NN literature, particularly regarding transition density function estimation,
is the limited analysis of estimation error and potential compromise of non-negativity.

In quantitative finance, many popular stochastic models have  unknown transition densities; however, their Fourier transforms, i.e.\ characteristic functions, are often explicitly available via the L\'{e}vy-Khintchine formula \cite{ken1999levy}. This property has been extensively utilized in option pricing through various numerical methods. Prominent among these are the Carr-Madan approach \cite{Carr1999}, the Convolution (CONV) technique \cite{lord2008}, Fourier Cosine (COS) method proposed by \cite{Fang2008}, Shannon-wavelet methods \cite{ortiz2016highly, dang2018dimension}, with the COS method being particularly noteworthy.
Specifically, the COS method is known for achieving high-order convergence for piecewise smooth problems.  However, within the broader framework of stochastic optimal control, where problems often exhibit complex and non-smooth characteristics, this high-order convergence is unattainable, as noted in \cite{LuDang2022, ForsythLabahn2017}.
{\siscblue{A notable drawback of the COS method is its lack of a mechanism to control the potential loss of non-negativity in estimated transition densities. This issue--more pronounced with short maturities--may lead to violations of the no-arbitrage principle in numerical value functions, posing significant challenges in stochastic optimal control where the accuracy of these values is crucial for making optimal decisions \cite{ForsythLabahn2017}.}} In the same vein of research, recent works on $\epsilon$-monotone Fourier methods for control problems in finance merit attention \cite{ForsythLabahn2017, online23, LuDang2022, LuDang2023}.

In response to the noted observations, this paper sets out to achieve three primary objectives. Firstly, we present a single-layer feed-forward (FF) NN approach to approximate transition densities with closed-form Fourier transforms, facilitating training in the Fourier domain.
This approach simplifies the implementation considerably when compared to deeper NN structures.
Second, we conduct a rigorous and comprehensive analysis of the $L_2$ estimation error between the exact and the estimated transition densities obtained through the proposed approach. This methodology, dubbed the \underline{Four}ier-trained Neural \underline{Net}work method or ``FourNet'', showcases the benefits of using the Fourier transform in FFNN models. Lastly, we validate FourNet's accuracy and versatility across a spectrum of stochastic financial models.  The main contributions of this paper are as follows.
\begin{itemize}

   \item
We establish two key results for FourNet: (i) transition densities can be approximated arbitrarily well
in the $L_2$-sense using a single-layer FFNN with a Gaussian activation function and a finite number of neurons; and (ii) the $L_2$-error in this approximation remains invariant under the Fourier transform map.
Here, $L_2\left(\mathbb{R}\right)$ denotes the space of square-integrable functions.

FourNet's methodology underscores the potential and efficacy of shallow NN architectures for complex approximation tasks. The inherent invariance under Fourier transformation opens opportunities for training and error analysis in the Fourier domain, rather than the conventional spatial domain.
This unique capability allows us to utilize the the known closed-form expression of the characteristic function and the Fourier transform of the FFNN for an in-depth analysis of the $L_2$ estimation error.

    \item
   Using FourNet, we formulate an approximation for transition densities using a single-layer FFNN equipped with a Gaussian activation function. FourNet's parameters are fine-tuned by minimizing a mean-squared-error (MSE) loss, supplemented with a mean-absolute-error (MAE) regularization. Both the loss function and regularization term stem from the known characteristic function and the Fourier transform of the FFNN. A strategic sampling approach is proposed, maximizing the benefits of MAE regularization.

We establish practical bounds for the \(L_2\) estimation error and potential loss of nonnegativity in FourNet {\siscb{for the
general case of $d$-dimensions ($d\ge 1$)}}, which are attributed to truncation, training, and sampling errors. These bounds highlight FourNet's advantages over existing Fourier-based and NN-based density estimation methods, offering valuable insights into its reliability and robustness in practical applications.

\item We showcase FourNet's accuracy and versatility across a broad spectrum of financial models, with a particular focus on option pricing. Our analysis encompasses the class of exponential L\'{e}vy processes, such as the CGMY model \cite{carr2002fine}, Merton jump-diffusion model \cite{merton1976option}, and Kou asymmetric double exponential model \cite{kou2002jump}, along with their multi-dimensional extensions.  We also explore the Heston model \cite{heston93} and its recent adaptations  that incorporate the self-exciting Queue-Hawkes jump process \cite{daw2022ephemerally, arias2022}. Notably, FourNet demonstrates exceptional robustness in handling ultra-short maturities and asymmetric heavy-tailed distributions, scenarios that pose significant challenges for traditional Fourier-based methods.

\end{itemize}
This paper introduces the FourNet method and its initial applications, setting the stage for subsequent works. Primarily focusing on European options, it also demonstrates FourNet's capabilities in pricing Bermudan options within L\'{e}vy process-based models. Future work will extend FourNet's application to complex stochastic control problems, including portfolio optimization, thereby broadening its scope and utility.
Although this work centers on  transition densities,
FourNet's methodology and its comprehensive error analysis are also relevant to the study of Green's functions
for parabolic integro-differential equations \cite{garronigreenfunctionssecond92}, due to their foundational relationship.

The remainder of the paper is organized as follows.
Section~\ref{section:FN} outlines the structure of single-layer FFNNs with non-sigmoid activation functions and introduces a related universal approximation theorem. Section~\ref{sc:fournet} presents FourNet, detailing key mathematical results and the MSE loss function. An error analysis of FourNet is detailed in Section~\ref{sc:error}.
Section~\ref{sc:Training} discusses sampling strategies, training considerations, and algorithms.  Section~\ref{sec:num} demonstrates FourNet's accuracy and versatility through extensive numerical experiments. The paper concludes in Section~\ref{sc:cc} with a discussion of potential future work.

\section{Background on single-layer FFNNs}
\label{section:FN}

\subsection{Non-sigmoid activation functions}
Feed-forward neural networks can be perceived as function approximators comprising of several inputs, hidden layers composed of neurons/nodes, an activation function, and several outputs.
This study primarily concentrates on shallow NNs characterized by a single input, a single output, and a number of nodes within the hidden layer. We also consider only the case that the input is one dimensional. Figure~\ref{fig:nn_diag} depicts a single-layer FFNN having a total of $N$ nodes in the hidden layer.

\begin{center}
\includegraphics[width=0.555\textwidth, height=0.22\textwidth]{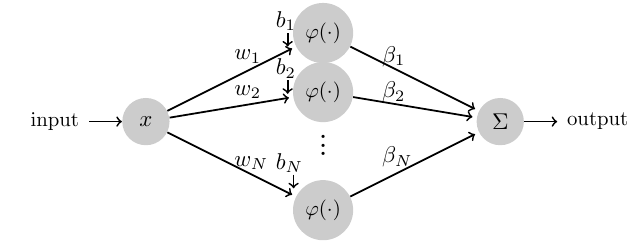}
\vspace*{-0.4cm}
\captionof{figure}{{\siscblue{Single-layer FFNN with one-dimensional input $x \in \mathbb{R}$, featuring $N$ neurons with weights $w_n$ and $\beta_n$, biases $b_n$, and activation function $\varphi(\cdot)$. }}}
\label{fig:nn_diag}
\end{center}

We now start with FFNNs with (Borel measurable) non-sigmoid activation functions, and the associated Universal Approximation Theorem \cite{Maxwell118640}[Theorem~2.1]. This class of of FFNNs is defined below.
\begin{definition}[$\Sigma^{\dagger}(\varphi)$ - activation function $\varphi$]
\label{def:Sigmaprime}
Let $\Sigma^{\dagger}(\varphi)$ be the class of single-layer FFNNs having arbitrary Borel measurable activation functions $\varphi$ defined by
\EQ
\label{eq:Sigmaprime}
\Sigma^{\dagger}(\varphi) = \bigg\{ \widehat{g} : \Rbb \to \Rbb \big{|} ~  \widehat{g}(x;\theta) = \sum_{n=1}^N  \beta_n  \varphi\l(w_n x + b_n\r),
~\beta_n, w_n, b_n \in \Rbb, ~ N \in \mathbb{N}\bigg\}.
\EN
Here, $x \in \mathbb{R}$ is the input; for a fixed $N$, the parameter $\theta \in \Rbb^{3N}$ is constituted by the weights $w_n$ and $\beta_n$,
and the bias terms $b_n$, $n = 1, \ldots, N$.
\end{definition}
For subsequent use, for $1 \le p < \infty$, we define the sets of $p$-integrable and $p$-locally-integrable  functions,
respectively denoted by $L_p\left(\mathbb{R}\right)$ and $L_p\left(\mathbb{R}, loc \r)$, as follows
\begin{linenomath}
\postdisplaypenalty=0
\begin{align}
\label{eq:Lp}
L_p\left(\mathbb{R}\right)
& = \bigg\{f \in \Mcal ~\big{|}~ \|f\|_p \equiv\bigg[\int |f(x)|^p \md x\bigg]^{1 / p}<\infty \bigg\},
\nonumber
\\
L_p\left(\mathbb{R}, loc \r)
&=
\bigg\{f \in \Mcal ~\big{|}~ f\, \Ibb_{[-A, A]} \in L_p\left(\mathbb{R}\right), \forall A \in \{1, 2, 3, \ldots\} ~ \bigg\}.
\end{align}
\end{linenomath}
Here, $\Mcal$ is the space of all Borel measurable functions $f : \Rbb \to \Rbb$.\footnote{{\siscblue{The set $\Mcal$ essentially contains all functions relevant to practical applications.}}}

Closeness of two elements $f_1$ and $f_2$ of $L_p\left(\mathbb{R}, loc\r)$ is measured by a metric $\rho_{p, \loc}(f_1, f_2)$ defined as follows \cite{Maxwell118640}
\EQ
\label{eq:rho}
\ds \rho_{p, \loc}\l(f_1, f_2\r)= \sum_{A = 1}^{\infty} \l(2^{-A}\r) \min \left(\left\|(f_1-f_2)~  \Ibb_{[-A, A]}\right\|_p, 1\right),\quad
f_1, f_2 \in L_p\left(\mathbb{R}, \loc \r).
\EN
{\siscblue{Here, the indicator function $\Ibb_{D}(\cdot)$ is defined as follows: $\Ibb_{D}(x) = 1$ if $x \in D$ and zero otherwise.}}
We now introduce the notion of $\rho_{p, \loc}$-denseness for $L_p(\mathbb{R}, \loc)$ \cite{Maxwell118640}.
\begin{definition}[$\rho_{p}$-denseness, $1\le p < \infty$]
\label{def:rho}
A subset $\Scal$ of $L_p\left(\mathbb{R}, \loc \r)$ is $\rho_{p, \loc}$-dense in $L_p\left(\mathbb{R}, \loc \r)$ if, for any $f_1$ in $L_p\left(\mathbb{R}, \loc \r)$ and any $\varepsilon>0$, there is a $f_2$ in $\Scal$ such that $\rho_{p, \loc}\l(f_1, f_2\r)<\varepsilon$,
where $\rho_{p, \loc}\l(f_1, f_2\r)$ is defined in \eqref{eq:rho}.
\end{definition}

\subsection{Universal Approximation Theorem}
The Universal Approximation Theorem proposed in \cite{hornik1989multilayer} for sigmoid activation functions
play a key theoretical foundation.  However, sigmoid activation functions are not necessary for universal approximation
as highlighted in \cite{Maxwell118640}[Theorem~2.1] - therein, an identical universal approximation theorem to the one in
\cite{hornik1989multilayer} was obtained.
The key finding of \cite{Maxwell118640} is that, for sufficiently complex single-layer FFNNs with an arbitrary (Borel measurable) activation function at the hidden layer can approximate an arbitrary target function $f(\cdot) \in L_p(\mathbb{R}, loc)$,  $1 \le p < \infty$,  arbitrary well, provided that the activation function, denoted by $\varphi(x)$, belong to $L_1(\mathbb{R}) \cap L_p(\mathbb{R})$ and
$\int_{\mathbb{R}} \varphi(x) dx$ does not vanish.  Formally, we state the Universal Approximation Theorem for non-sigmoid activation functions  below.
\begin{theorem}[Universal Approximation Theorem \cite{Maxwell118640}~Theorem~2.1]
\label{theorem:u2}
Let  $\varphi$ be the (Borel measurable) activation function that belongs to $L_1(\Rbb) \cap L_p(\Rbb)$, $1\le p< \infty$.
If $\int_{\Rbb} \varphi(x)~ dx \neq 0$, then $\Sigma^{\dagger}(\varphi)$ is $\rho_{p, \loc}$-dense in  $L_p (\Rbb, loc)$.
Here, $\Sigma^{\dagger}(\varphi)$ and $\rho_{p, \loc}$ are respectively defined in Definitions~\ref{def:Sigmaprime} and \ref{def:rho}.
\end{theorem}

\section{A Fourier-trained network (FourNet)}
\label{sc:fournet}
We denote by $T>0$ a finite horizon, and let $t$ and $\Delta t$ be fixed such that $0\le t < t + \Delta t \le T$.
For the sake of exposition, we focus on estimating a time and spatially homogeneous transition density, denoted by $g(\cdot)$ and is represented as $g(x, t + \Delta t;y, t ) = g(x- y;\Delta t)$. Such transition densities are characteristic of L\'{e}vy processes.

{\siscblue{Our methodology also applies to models like the Heston and Heston-Queue-Hawkes,
which exhibit non-homogeneity in variance. For these models, pricing European options needs only a single training session; however, control problems typically require multiple independent sessions to address time-stepping.
Each session uses a dataset derived from characteristic functions conditioned on different starting and ending variance values.
These sessions can be conducted in parallel, thereby enhancing computational efficiency.
This aligns with existing Fourier methods
(e.g.\ \cite{fang2011fourier, ForsythLabahn2017, ruijter2012two}) which similarly utilize characteristic functions conditioned on variance values at each timestep. We plan to explore Heston-type models in control problems in future work.}}

For notational simplicity, we momentarily suppress the explicit dependence of the transition density on $\Delta t$,
denoting $g(\cdot) \equiv g(\cdot; \Delta t): \Rbb \to \Rbb$ as the transition density we seek to approximate using single-layer FFNNs.
The importance of $\Delta t$ will be highlighted in our applications detailed in Section~\ref{sec:num}.

Since the transition density $g(\cdot)$ is almost everywhere bounded on $\Rbb$, together with the fact that $g \in L_1(\Rbb)$, we have $g \in L_2(\Rbb)$. Therefore, we consider approximating a transition density $g \in L_1(\Rbb) \cap L_2(\Rbb)$.
\subsection{A universal approximation result in  $\boldsymbol{L_2(\Rbb)}$}
We now provide a refined version of the Universal Approximation Theorem~\ref{theorem:u2} for target functions in $L_2(\Rbb)$. Specifically, by invoking H\"{o}lder's inequality, we have that $L_2\left(\mathbb{R}\right) \subset L_2\left(\mathbb{R}, \text{loc}\r)$.
A natural question thus emerges: if the function we aim to approximate, $f$, belongs to $L_2 (\Rbb)$ rather than
$L_2 (\Rbb, loc)$, can we identify an FFNN in $\Sigma^{\dagger}(\varphi) \cap L_2 (\Rbb)$ that approximates $f$ arbitrarily well,  in the sense of the Universal Approximation Theorem~\ref{theorem:u2}?
In the forthcoming lemma, we affirmatively address this question.

\begin{lemma}[$\rho_{2, \loc}$-denseness of $\Sigma^{\dagger}\l(\varphi\r) \cap L_2 (\Rbb)$]
\label{lemma:g_rho}
Let $\varphi$ be a continuous activation function that belongs to $L_1 (\Rbb) \cap L_2 (\Rbb)$. Assume that $f(\cdot)$ is in $L_2 (\Rbb)$.
For any $\epsilon > 0$, there exists a neural network $f'(\cdot; \theta') \in \Sigma^{\dagger}\l(\varphi\r) \cap L_2 (\Rbb)$
such that $\rho_{2, \loc}\l(f, f'\r) < \epsilon$, where $\rho_{2, \loc}(\cdot)$ is defined in Definition~\ref{eq:rho}.
\end{lemma}


\begin{proof}[Proof of Lemma~\ref{lemma:g_rho}]
Since $\varphi \in L_1(\Rbb) \cap L_2(\Rbb)$, it satisfies the conditions of Theorem~\ref{theorem:u2}
for $p = 2$. Therefore, $\exists f'(\cdot; \theta') \in \Sigma^{\dagger}\l(\varphi\r)$ such that
$\rho_{2, \loc}\l(f, f'\r) < \eps$.
Here, $f'(x; \theta')= \sum_{n=1}^{N'}  \beta_n'  \varphi(w_n' x + b_n')$,
where $N'$ is the finite number of neurons.
As $f'(\cdot; \theta') \in L_2(\mathbb{R}, loc)$, each $\beta_n'  \varphi(w_n' x + b_n')$
is square-integrable on any compact set.
Since $\varphi \in L_2(\Rbb)$, the square-integrability of $\beta_n'  \varphi(w_n' x + b_n')$ implies that $|\beta_n'| <\infty$ for all $n \le N'$. Finally, as $f'(\cdot; \theta')$ is a finite sum of functions in $L_2(\mathbb{R})$,
it follows that $f'(\cdot; \theta')\in L_2(\mathbb{R})$. This concludes the proof.
\end{proof}

\subsection{Gaussian activation function $\boldsymbol{e^{-x^2}}$}
Building upon Lemma~\ref{lemma:g_rho}, we present a corollary focusing on the Gaussian activation function $\varphi(x) \equiv
\phi(x) = e^{-x^2}$.
\begin{corollary}
\label{col:gauss}
For a target function $f(\cdot)$ in $L_2(\Rbb)$ and any $\epsilon > 0$, there exists an FFNN $f'(\cdot; \theta') \in \Sigma^{\dagger}(\phi) \cap L_2(\Rbb)$ with $\phi(x) = e^{-x^2}$, where $f'$ approximates $f$ such that $\rho_{2, \loc}(f, f') < \epsilon$. The FFNN $f'(\cdot; \theta')$ has bounded parameters:
$0 < |\beta_n'| < \infty $, $0 < |w_n'| < \infty$, and $|b_n'| < \infty$, $\forall n = 1, \ldots, N'$.
\end{corollary}
\begin{proof}
Through integration, we verify that  $\phi(x) = e^{-x^2}\in L_1(\Rbb) \cap L_2(\Rbb)$.
Thus, by Lemma~\ref{lemma:g_rho}, $\exists$ $f'(\cdot; \theta') \in \Sigma^{\dagger}(\phi) \cap L_2 (\Rbb)$
such that $\rho_{2, \loc}(f, f') < \epsilon$. Here, $f'(x; \theta')= \sum_{n=1}^{N'}  \beta_n'  \varphi(w_n' x + b_n')$.
By Lemma~\ref{lemma:g_rho}, $|\beta_n'|~<~\infty$, $\forall n \le N'$.
Additionally, it must be true that $\beta_n' \neq 0$, $\forall n\le N$;
otherwise the corresponding neuron output is zero.
For the same reason, we also have $|w_n'|, |b_n'| < \infty$, $\forall n\le N'$.
Lastly, it is also the case that  $w_n' \neq 0$ for all $n \le N'$. Otherwise, suppose that  $w_{k}' =  0$ for $n = k\le N'$,
then  $0<|\beta_{k}'| \exp(-(w_{k}' x + b_{k}')^2) = |\beta_{k}'| \exp(-(b_{k}')^2) \le c$, where  $0<c<\infty$, contradicting with $f'(\cdot; \theta') \in L_2 $. This concludes the proof.
\end{proof}
Corollary~\ref{col:gauss} establishes that for any target function in $L_2(\mathbb{R})$, an FFNN with bounded parameters exists within $\Sigma^{\dagger}(\phi) \cap L_2(\Rbb)$, where $\phi(x) = e^{-x^2}$, capable of approximating the target function arbitrarily well as measured by $\rho_{2, \loc}(\cdot, \cdot)$.
This result allows us to narrow down our focus to $\Sigma(\phi)$, a more specific subset of $\Sigma^{\dagger}(\phi)$,
that encapsulates FFNNs characterized by parameter bounds.
We now formally define $\Sigma(\phi)$ and its associated bounded parameter space $\Theta$,
both of which are crucial for our subsequent theoretical analysis and practical application:
\vspace*{-0.1cm}
\begin{align}
\label{eq:Sigma}
\Sigma(\phi) = \big\{ \widehat{g} : \Rbb \to \Rbb \big{|} ~  \widehat{g}(x;\theta) = \sum_{n=1}^{N}  \beta_n  \phi(w_n x + b_n), ~\phi(x) = e^{-x^2}, ~\theta \in \Theta\big\},&
\\
\Theta = \big\{\theta \in \Rbb^{3N}~ \big{|}~0<|\beta_n| < \infty,~0< |w_n|<\infty,~|b_n|<\infty,~
n = 1, \ldots, N\big\}.&
\label{eq:Theta}
\end{align}

\vspace*{-0.25cm}
\begin{remark}
\label{rm:l2}
All FFNNs in $\Sigma(\phi)$, where $\phi(x) = e^{-x^2}$, belong to $L_2(\Rbb)$ as per definitions \eqref{eq:Sigma} and~\eqref{eq:Theta}. By Corollary~\ref{col:gauss},
given any target function in $f(\cdot)$ in $L_2(\Rbb)$ and any $\epsilon > 0$, there exists an FFNN $f'(\cdot; \theta') \in \Sigma(\phi)$ such that $\rho_{2, \loc}(f, f') < \epsilon$.
\end{remark}

\subsection{Existence and invariance of FourNet}
We now establish a key result demonstrating the existence of $\widehat{g}(\cdot; \theta_{\eps}^*) \in \Sigma\l(\phi\r)$,
where $\Sigma\l(\phi\r)$ is defined in \eqref{eq:Sigma}, that is capable of approximating the exact transition density $g(\cdot)$
arbitrarily well in the $L_2$-sense. We hereafter refer to $\widehat{g}(\cdot; \theta_{\eps}^*)$ as a theoretical  FFNN approximation to the true transition density $g(\cdot)$. Furthermore, we also show that the associated theoretical approximation error in $L_2$ remains invariant under the Fourier transform map.

To this end, we recall that the transition density $g(\cdot)$ and the associated characteristic function  $G(\eta)$
are a Fourier transform pair.  They are defined as follows
\EQS
\label{eq:small_g}
\mathfrak{F}[g(\cdot)](\eta) \equiv   G(\eta) = 
\int_{-\infty}^{\infty} e^{i \eta x}\, g(x)\, dx,
~~
\mathfrak{F}^{-1}[G(\cdot)](x) \equiv g(x) =  
\f{1}{2\pi} \int_{-\infty}^{\infty} e^{-i \eta x}\, G(\eta)\, d\eta.
\ENS
For subsequent discussions,  for a complex-valued function $f:\Rbb \to \mathbb{C}$, we denote by $\text{Re}_f(\cdot)$ and $\text{Im}_f(\cdot)$ its real and imaginary parts. {\siscblue{We also have}} $\l|f(\cdot)\r|^2 = f(\cdot)\overline{f(\cdot)}$,
where $\overline{f(\cdot)}$ is the complex conjugate of $f(\cdot)$.

We will also utilize the Plancherel Theorem, which is sometimes also referred to as the Parseval-Plancherel identity \cite{yosida1968functional, anker2006lie, kobayashi2011representation}. For the sake of convenience, we reproduce it below.
Let  $f: \mathbb{R} \to \mathbb{R}$ be a function in $L_1(\mathbb{R}) \cap L_2(\mathbb{R})$.
The Plancherel Theorem states that its Fourier transform $\mathfrak{F}[f(\cdot)](\eta)$ is in $L_2(\mathbb{R})$, and
\EQ
\label{eq:Plancherel}
\int_{\Rbb} \l|f(x)\r| ^2 \md x = \frac{1}{2\pi}\int_{\Rbb} \l|\mathfrak{F}[f(\cdot)](\eta)\r| ^2 \md \eta.
\EN

\begin{theorem}[FourNet's existence result]
\label{theorem:g_G_eps1}
Given any $\epsilon>0$, there exists an FFNN $\widehat{g}(\cdot; \theta_{\eps}^*)\in \Sigma(\phi)$, where $\Sigma(\phi)$ is defined in \eqref{eq:Sigma}, that satisfies the following 
\EQ
\label{eq:Gbound}
  \int_{\Rbb} \l|g(x)- \widehat{g}(x; \theta_{\eps}^*)\r|^2~\md x = \frac{1}{2\pi} \int_{\Rbb} \l|G(\eta)- \widehat{G}(\eta; \theta_{\eps}^*)\r|^2~\md \eta
 ~<~ \epsilon.
\EN
Here, $\widehat{G}(\eta; \theta_{\eps}^*)$ is the Fourier transform of $\widehat{g}(\cdot; \theta_{\eps}^*)$, i.e.\ $\widehat{G}(\eta; \theta_{\eps}^*) = \mathfrak{F}\l[\widehat{g}(\cdot; \theta_{\eps}^*)\r](\eta)$.
\end{theorem}
\begin{proof}[Proof of Theorem~\ref{theorem:g_G_eps1}]
We first show $\int_{\Rbb} \l|g(x)- \widehat{g}(x; \theta_{\eps}^*)\r|^2~\md x < \epsilon$, then
the equality in \eqref{eq:Gbound}.
Since $g(\cdot)$ and $\widehat{g}(\cdot; \theta_{\eps}^*)$ are in $L_2(\Rbb)$, $\exists A'$ sufficiently large such that
\begin{linenomath}
\postdisplaypenalty=0
\begin{align}
\int_{\Rbb \setminus [-A',A']} |g(x)|^2 ~\md x~ < ~\epsilon/8,
\quad
\int_{\Rbb \setminus [-A',A']} |\widehat{g}(x; \theta_{\eps}^*)|^2 ~\md x ~< ~\epsilon/8.
\label{eq:ep1_18}
\end{align}
\end{linenomath}
{\siscblue{By Remark~\ref{rm:l2}, there exists $\widehat{g}(x; \theta_{\eps}^*) \in \Sigma\l(\phi\r)$ such that}}
\EQS
\ds \rho_{2, \loc}\l(g, \widehat{g}\l(\cdot; \theta_{\epsilon}^*\r)\r)
=
\sum_{A= 1}^{\infty} 2^{-A} \min \left(\left\|(g-\widehat{g}(\cdot; \theta_{\eps}^*))~  \Ibb_{[-A, A]}\right\|_2, 1\right)
< \frac{\epsilon^{1/2}}{2^{1/2}}2^{-A'}.
\ENS
Therefore, 
\EQS
\label{eq:rho_2_loc}
\ds 2^{-A'} \min \left(\left\|(g-\widehat{g}(\cdot; \theta_{\eps}^*))~  \Ibb_{[-A', A']}\right\|_2, 1\right)
< \frac{\epsilon^{1/2}}{2^{1/2}}2^{-A'},
\ENS
from which, we have
\EQ
\label{eq:ep1_12}
\int_{[-A',A']} \l(g(x)- \widehat{g}(x;\theta_{\eps}^*)\r)^2 ~\md x < \epsilon/2.
\EN
Using \eqref{eq:ep1_18}-\eqref{eq:ep1_12}, we have $\ds \int_{\Rbb} (g(x)-\widehat{g}(x;\theta_{\eps}^*))^2 ~\md x = \ldots$
\begin{linenomath}
\postdisplaypenalty=0
\begin{align}
\ldots &= \int_{[-A',A']} \l|g(x)-\widehat{g}(x;\theta_{\eps}^*)\r|^2 ~dx
+ \int_{\Rbb \setminus [-A',A']} \l|g(x)-\widehat{g}(x;\theta_{\eps}^*)\r|^2 ~\md x\\
        \nonumber
        & \qquad \qquad < \epsilon/2 + \int_{\Rbb \setminus [-A',A']} 2\l|g(x)^2+ \widehat{g}(x;\theta_{\eps}^*)^2\r|~ \md x
        < \epsilon,
\end{align}
\end{linenomath}
as wanted. Next, the equality in \eqref{eq:Gbound} follows directly from the Plancherel Theorem \eqref{eq:Plancherel},
noting $L_1(\mathbb{R})$ and $L_2(\mathbb{R})$ are closed under addition.
This completes the proof.
\end{proof}

\vspace*{-0.25cm}
\begin{remark}
\label{rm:sig}
Theorem~\ref{theorem:g_G_eps1} presents a significant theoretical result, demonstrating that the FourNet can approximate the exact transition density $g(\cdot)$ within an error of any given magnitude in the $L_2$-sense.\footnote{This result is expected since Gaussians with fixed variance are dense in $L^2(\Rbb)$ \cite{calcaterra2008approximating}.}  Interestingly, this error is invariant under the Fourier transform, tying together FourNet's approximation capabilities in both spatial and Fourier domains. This invariance opens opportunities for training and error analysis in the Fourier domain instead of the spatial domain.
In particular, it enables us to utilize the known closed-form expression of the characteristic function $G(\cdot)$, a process we elaborate on in subsequent sections.
\end{remark}

\subsection{Loss function}
\label{sc:FourierNN}
Recall that $\widehat{g}(x; \theta)$ in $\Sigma(\phi)$ has the form
\EQ
\label{eq:approx_nn*}
\widehat{g}(x; \theta) = \sum_{n =1}^{N} \beta_n \phi\l(w_{n} x + b_n\r), \quad \phi(x) = \exp(-x^2), \quad \theta \in \Theta.
\EN
We let $\widehat{G}(\cdot; \theta)$ be the Fourier transform of $\widehat{g}(\cdot; \theta)$, i.e.\ $\widehat{G}(\eta; \theta) = \mathfrak{F}\l[\widehat{g}(\cdot; \theta)\r](\eta)$. {\siscblue{By substitution, we have
\begin{align}
\label{eq:Ghat1}
\widehat{G}(\eta; \theta) &= \int e^{i\eta x} \widehat{g}(x; \theta)~dx
= \sum_{n=1}^N \beta_n \int e^{i\eta x} \phi (w_n x + b_n)~dx
\nonumber
\\
&= \sum_{n=1}^N \beta_n \int_{\Rbb} \cos (\eta x)~\phi (w_n x + b_n)~dx +i \sum_{n=1}^N \beta_n \int_{\Rbb} \sin(\eta x)~\phi (w_n x + b_n)~dx
\nonumber
\\
&=\text{Re}_{\widehat{G}}(\eta; \theta) + i\text{Im}_{\widehat{G}}(\eta; \theta).
\end{align}
Here, by integrating the integral terms with $\phi(x) = \exp(-x^2)$, we obtain }}
\begin{align*}
\text{Re}_{\widehat{G}}(\cdot) \!=\!\!\sum_{n=1}^N\! \frac{\beta_{n}\sqrt{\pi}}{w_{n}} \cos\bigg(\frac{\eta b_n}{w_{n}}\bigg)\! \exp\bigg(\frac{-\eta^2}{4w_{n}^2}\bigg),
~\text{Im}_{\widehat{G}}(\cdot) \!= \!\!\sum_{n=1}^N\! \frac{\beta_{n}\sqrt{\pi}}{w_{n}} \sin\bigg(\frac{-b_n \eta}{w_{n}}\bigg)\! \exp\bigg(\frac{- \eta^2}{4w_{n}^2}\bigg).
\end{align*}
\noindent Recall that our starting point is that $G(\cdot)$, the Fourier transform of the transition density $g(\cdot)$, is known in closed form.
Therefore, motivated by Theorem~\ref{theorem:g_G_eps1}, the key step of our methodology is to use the known data $\{\l(\eta, \text{Re}_G(\eta; \theta)\r)\}$ and
$\{\l(\eta, \text{Im}_G(\eta; \theta)\r)\}$ to train  $\widehat{G}(\eta; \theta)$ using the expressions in \eqref{eq:Ghat1}.

To this end, we restrict the domain of $\eta$ from $\Rbb$ to a fixed finite interval $[-\eta', \eta']$, where $0<\eta'<\infty$ and is sufficiently large. We denote the total number of training data points by $P$, and we consider a deterministic, potentially non-uniform, partition $\{\eta_p\}_{ p = 1}^P$ of the interval $[-\eta', \eta']$.
With  $\delta_p = \eta_{p+1}-\eta_p$, $p = 1, \ldots, P-1$,
we assume
\EQ
\label{eq:eta}
\delta_{\min} = C_0/P, \quad \delta_{\max} = C_1/P, \quad
{\text{with $\delta_{\min} = \min_p \delta_p$ and $\delta_{\max} = \max_p \delta_p$,}}
\EN
where the constants $C_0, C_1>0$ are finite and  independent of $P$. Letting $\widehat{\Theta} \subseteq \Theta$ be the empirical parameter space, we introduce an empirical loss function ${\text{Loss}}_{P}(\theta)$, $\theta \in \widehat{\Theta}$, below
\begin{linenomath}
\postdisplaypenalty=0
\begin{align}
\label{eq:loss}
{\text{Loss}}_P(\theta) &=
\frac{1}{P} \sum_{p=1}^{P} \l|G(\eta_p) - \widehat{G}(\eta_p; \theta)\r|^2 + R_{P}(\theta),~~
{\text{$\{\eta_p\}_{ p = 1}^P$ satisfying \eqref{eq:eta}}}.
\end{align}
\end{linenomath}
Here, $\widehat{G}(\eta_p;\theta)$ is defined in \eqref{eq:Ghat1}, with $R_{P}(\theta)$ as the MAE regularization term
\EQ
\label{eq:reg}
R_{P}(\theta) = \frac{1}{P }\sum_{p=1}^{P}
\big(|{\siscblue{\text{Re}_{G}(\eta_p)}} - \text{Re}_{\widehat{G}}(\eta_p; \theta)|
+  |{\siscblue{\text{Im}_{G}(\eta_p)}} - \text{Im}_{\widehat{G}}(\eta_p; \theta)|\big).
\EN
By training ${\text{Loss}}_P(\cdot)$, we aim to find the empirical minimizer $\widehat{\theta}^* \in \widehat{\Theta}$, where
\EQ
\label{eq:thetastar}
\widehat{\theta}^* = \arg\min_{\theta \in \widehat{\Theta}} {\text{Loss}}_{P}(\theta).
\vspace*{-0.5cm}
\EN
{\siscblue{\begin{remark}[MAE regularization]
\label{rm:mse}
The incorporation of the MAE regularization term in the loss function \eqref{eq:loss} is strategically motivated by its ability to significantly enhance FourNet's robustness and accuracy through two key mechanisms listed below.
\begin{itemize}
\item Control over non-negativity: As detailed in Remark~\ref{rm:non-neg},  the MAE component  enables direct control over the upper bound of the potential pointwise loss of non-negativity in $\widehat{g}(\cdot; \theta)$. This control over non-negativity though training of the loss function not only enhances the mathematical integrity of our density estimates but also underscores FourNet's significant practical advantages over traditional Fourier-based and NN-based density estimation methods.

\item
Enhanced training accuracy in critical regions: The MAE component improves accuracy specifically at critical regions identified through deterministic partition points ${\eta_p}_{p=1}^P$. These points target critical areas, such as those with convexity changes and peaks, in both the real ($\text{Re}_G(\cdot)$) and imaginary ($\text{Im}_G(\cdot)$) components of $G(\cdot)$.
This focus ensures precise local fits,  thereby complementing \(L_2\)-error minimization.
This strategy not only aims for an optimal overall fit across the entire domain of \(G(\cdot)\) but also prioritizes precision at points critical to FourNet's effectiveness. It aligns with our overarching goal of minimizing \(L_2\)-errors, crucial for subsequent \(L_2\)-error analysis.  Further details on the selection of $\{\eta_p\}_{ p = 1}^P$ are discussed in Subsection~\ref{ssc:sampling}
\end{itemize}
\end{remark}}}
We conclude that, for deep NNs, the function \(\widehat{g}(\cdot; \theta)\) is expressed as a composition of functions.
However, computing its Fourier transform can be very complex, as noted by \cite{bergner2006spectral}. Yet, our extensive numerical experiments have demonstrated that a single-layer FFNN possesses remarkable estimation capabilities.

\begin{remark}[Truncation error in the Fourier domain]
\label{rm:trunc}
Note that,  given the boundedness of the parameter space $\Theta$,
both $|\text{Re}_{\widehat{G}}(\cdot; \theta)|$ and $|\text{Im}_{\widehat{G}}(\cdot; \theta)|$, $\theta \in \Theta$,
are in $L_1 (\Rbb)$. We also recall that both $|\text{Re}_{G}(\cdot)|$ and $|\text{Im}_{G}(\cdot)| \in L_1(\Rbb)$.
Furthermore, since both $g(\cdot)$ and $\widehat{g}(\cdot; \theta)\in \Sigma(\phi)$, for any $\theta \in \Theta$,
are in $L_1(\Rbb)\cap L_2(\Rbb)$, by the Plancherel Theorem~\eqref{eq:Plancherel},
both $|G(\cdot)|$ and $|\widehat{G}(\cdot;\theta)|$ are in $L_2(\Rbb)$.
Therefore, for any given $\epsilon>0$,  $\exists \eta'>0$ such that, with $f \in \{\text{Re}_{G}, \text{Im}_{G}, \text{Re}_{\widehat{G}}(\cdot; \theta),
\text{Im}_{\widehat{G}}(\cdot; \theta)\}$ and $h \in \{G(\cdot), \widehat{G}(\cdot; \theta)\}$,
\begin{linenomath}
\postdisplaypenalty=0
\begin{align}
\int_{\Rbb\setminus[-\eta',\eta']} |f(\eta)|~d\eta < \epsilon, \quad
\int_{\Rbb\setminus[-\eta',\eta']} |h(\eta)|^2~d\eta < \epsilon,
\quad \forall \theta \in \Theta.
\label{eq:Rb2}
\end{align}
\end{linenomath}
That is, the truncation error in the Fourier domain can be made arbitrarily small
by choosing $\eta'>0$ sufficiently large. In practice, given
a closed-form expression for $G(\cdot)$, $\eta'$ can be determined numerically,
as illustrated in Subsection~\ref{ssc:sample}.
\end{remark}

\begin{remark}[Gaussian mixtures]
\label{rm:GM}
There are two potential interpretations of our methodology. The first  sees $\widehat{g}(x; \theta)$ in \eqref{eq:approx_nn*}
as an FFNN approximation of the exact transition density $g(\cdot)$, and its
parameters are learned by minimizing the loss function ${\text{Loss}}_{P}(\theta)$ (defined in \eqref{eq:loss}).
Alternatively, $\widehat{g}(x; \theta)$ in \eqref{eq:approx_nn*} can be written as
\EQ
\label{eq:Gaussian_mix}
\widehat{g}(x; \theta) = \sum_{n = 1}^N  \frac{1}{\sqrt{2\pi \sigma_n^2}} \exp\l(- \frac{(x - \mu_n)^2}{2\sigma_n^2}\r),
\quad \mu_n = -\frac{b_n}{w_n},~ \sigma_n^2 = \frac{1}{2w_n^2}.
\EN
This can be essentially viewed as a Gaussian mixture with $N$ components \cite{mclachlan2019finite}, where the $n$-th Gaussian component has mean $\mu_n = -\frac{b_n}{w_n}$ and variance $\frac{1}{2w_n^2}$. Unlike traditional Gaussian mixtures, the centers of the component distributions are not predetermined but are also learned through training.
Finally, it is worth noting that the set of all normal mixture densities is dense in the set of all density functions under the $L_1$-metric (see \cite{li1999mixture}), hence a mixture of Gaussian  like in \eqref{eq:Gaussian_mix} can be used to estimate any unknown density function.
\end{remark}

\section{Error analysis}
\label{sc:error}
We denote by $\widehat{\theta}$ the parameter learned from training the loss function ${\text{Loss}}_{P}(\theta)$,
and refer to $\widehat{g}(\cdot; \widehat{\theta})$ as the corresponding estimated transition density.
We aim to derive an upper bound for the   $L_2$ estimation error $\int_{\Rbb} |g(x)- \widehat{g}(x; \widehat{\theta})|^2~ dx$.
By the Plancherel Theorem~\eqref{eq:Plancherel}, we have
$\int_{\Rbb} |g(x)- \widehat{g}(x; \widehat{\theta})|^2~ dx =
\int_{\Rbb} \big|G(\eta) - \widehat{G}(\eta; \widehat{\theta})\big|^2 d\eta$.
This underscores the unique advantages of the proposed approach: error analysis is better suited to the Fourier domain than to the spatial domain, as we can directly benefit from the loss function ${\text{Loss}}_{P}(\theta)$, which is designed specifically for this domain.

In our error analysis, we require $C' := \sup_{\eta, \theta}|\partial |G(\eta) -  \widehat{G}(\eta; \theta)|^2 /\partial \eta| < \infty$,
for all $\eta\in [-\eta', \eta']$ and \(\theta \in \Theta \).
Given that \(\Theta\) is bounded and thus \(\widehat{G}(\eta; \theta)\) possesses a bounded first derivative, the requirement for \(C' < \infty\) is that \(G(\eta)\) also has a bounded first derivative. This leads us to the assumption that the random variable associated with the density \(g(\cdot)\) is absolutely integrable.
We also recall $C_0$ and  $C_1$ from \eqref{eq:eta}.
We now present an error analysis of the FourNet method in Lemma~\ref{lemma:Gtog} below.
\begin{lemma}
\label{lemma:Gtog}
As per Remark~\ref{rm:trunc}, for a given $\epsilon_1>0$, let the truncated Fourier domain $[-\eta', \eta']$ be such that,  with $f \in \{\text{Re}_{G}(\cdot), \text{Im}_{G}(\cdot), \text{Re}_{\widehat{G}}(\cdot; \theta), \text{Im}_{\widehat{G}}(\cdot; \theta)\}$ and $h~\in~\{G(\cdot), \widehat{G}(\cdot; \theta)\}$,
\begin{linenomath}
\postdisplaypenalty=0
\begin{align}
\int_{\Rbb\setminus[-\eta',\eta']} |f(\eta)|~d\eta < \epsilon_1, \quad
\int_{\Rbb\setminus[-\eta',\eta']} |h(\eta)|^2~d\eta < \epsilon_1,
~ \forall \theta \in \Theta.
\label{eq:Rb}
\end{align}
\end{linenomath}
Suppose that the parameter $\widehat{\theta}$ learned by training the empirical loss function ${\text{Loss}}_{P}(\theta)$, $\theta \in \widehat{\Theta}$, defined in \eqref{eq:loss},
is such that
\EQ
\label{eq:loss_eps}
    \bigg|\frac{1}{P }\sum_{p=1}^{P}\big|G(\eta_p)-\widehat{G}(\eta_p; \widehat{\theta})\big|^2+ R_{P}(\widehat{\theta}) \bigg|<\epsilon_2,
\EN
and the regularization term $R_{P}(\widehat{\theta})<\epsilon_3$, where $\epsilon_2, \epsilon_3 >0$.
Then, we have
\EQ
\label{eq:g_eps}
   \int_{\Rbb} \big|g(x)-\widehat{g}(x;\widehat{\theta})\big|^2 ~dx <
 \frac{1}{2\pi}\left(4\epsilon_1 + C_1( \epsilon_2 +\epsilon_3) +\frac{C'C_1^2}{P}\right).
\EN
\end{lemma}
\begin{proof}[Proof of Lemma~\ref{lemma:Gtog}]
Applying the error bound for the composite left-hand-side quadrature rule on a non-uniform partition gives
%
\begin{align}
\label{eq:Pbound}
\big|\sum_{p=1}^{P} \delta_p \big|G(\eta_p)-\hat{G}(\eta_p; \theta)\big|^2 \!\!-\!\! \int_{[-\eta',\eta']}\!\!\!\! \!\big|G(\eta)-\hat{G}(\eta; \theta)\big|^2~ d\eta\big|
\le  C'P (\delta_{\max})^2 = C/P.
\end{align}
noting $\delta_{\max} = C_1/P$, as in \eqref{eq:eta}, where $C = C'C_1^2$.
{\siscblue{Therefore,
\EQA
\label{eq:GboundP}
\int_{\Rbb} \big|G(\eta)&-&\widehat{G}(\eta; \widehat{\theta})\big|^2~ d\eta \overset{\text{(i)}}{=}
\int_{\Rbb\setminus[-\eta',\eta']}\!\!\!\! \big|G(\eta)-\widehat{G}(\eta; \widehat{\theta})\big|^2~ d\eta
+
\int_{[-\eta',\eta']}\!\! \!\!\!\big|G(\eta)-\widehat{G}(\eta; \widehat{\theta})\big|^2~ d\eta
\nonumber
\\
    && \overset{\text{(ii)}}{<} 4\epsilon_1   + \frac{C_1}{P}\sum_{p=1}^{P}  \l|G(\eta_p)-\widehat{G}(\eta_p; \widehat{\theta})\r|^2 + C/P
    \nonumber
    \\
    &&   \overset{\text{(iii)}}{<}4\epsilon_1+ C_1\big(\epsilon_2 +\epsilon_3\big) + C/P.
\ENA
Here, from (i) to (ii), we respectively bound the first and the second terms in (i)  by $4\epsilon_1$, using \eqref{eq:Rb}, and by $C/P+\sum_{p=1}^{P}\delta_{\max} \l|G(\eta_p)-\widehat{G}(\eta_p; \widehat{\theta})\r|^2$,
using \eqref{eq:Pbound} with $\theta = \widehat{\theta}$, noting that
$\delta_{p}\le \delta_{\max} = C_1/P$ $\forall p$. In (iii), we use \eqref{eq:loss_eps} and $R_{P}(\widehat{\theta})<\epsilon_3$ to bound the second term in (ii) by  $C_1(\epsilon_2 +\epsilon_3).$
Using the Plancherel Theorem \eqref{eq:Plancherel} and \eqref{eq:GboundP} gives
\begin{equation*}
   2\pi  \int_{\Rbb} \big|g(x)-\widehat{g}(x;\widehat{\theta})\big|^2 ~dx
    <    4\epsilon_1 + C_1( \epsilon_2 +\epsilon_3) +\frac{C'C_1^2}{P},
\end{equation*}
noting  $C = C'C_1^2$. Rearrange the above gives \eqref{eq:g_eps}.
This completes the proof.}}
\end{proof}
Lemma~\ref{lemma:Gtog}[Eqn.\ \eqref{eq:g_eps}] decomposes the upper bound for the $L_2$ estimation error into several  error components listed below.
\begin{itemize}
 \item Truncation error (Fourier domain): This arises from truncating the sampling domain from \( \Rbb \) to \( [-\eta', \eta'] \).
 It is bounded by $\epsilon_1$ (see \eqref{eq:Rb}), contributing $4\epsilon_1/(2\pi)$ to the derived error bound in \eqref{eq:g_eps}.

\item {\siscblue{Training error: This error results from the deviation of the learned parameters \(\widehat{\theta}\) from minimizing the empirical loss function, which includes the MSE  error and MAE regularization. Its total contribution to the error bound in \eqref{eq:g_eps} is \(C_1(\epsilon_2 + \epsilon_3)/(2\pi)\).}}

\item {\siscblue{Sampling error: This is caused by the use of a finite set of $P$ data points in training.  This error, captured as numerical integration error, is represented as $\frac{C'C_1^2}{2\pi P}$ in the error bound in \eqref{eq:g_eps}.}}
\end{itemize}
\begin{remark}[Nonnegativity of $\widehat{g}(\cdot;\widehat{\theta})$.]
\label{rm:non-neg}
We now investigate the potential loss of nonnegativity in $\widehat{g}(\cdot;\widehat{\theta})$,
where $\widehat{\theta}$ is learned as per Lemma~\ref{lemma:Gtog}.
To this end,  we use \( |\min(\widehat{g}(x; \widehat{\theta}), 0)| \), for an arbitrary $x \in \Rbb$,
as a measure of this potential pointwise loss.
Following similar steps (i)-(ii) of \eqref{eq:Pbound}, noting $R_{P}(\widehat{\theta})<\epsilon_3$,
we have
\EQAS
\int_{\Rbb} \big(|\text{Re}_{G}(\eta)-\text{Re}_{\widehat{G}}(\eta; \widehat{\theta})|+ |\text{Im}_{G}(\eta)-\text{Im}_{\widehat{G}}(\eta; \widehat{\theta})|\big)~d\eta
< 4\epsilon_1 + C_1\epsilon_3 + C'C_1^2/P.
\ENAS
Hence, $|\min(\widehat{g}(x; \widehat{\theta}), 0)| \le |g(x) - \widehat{g}(x; \widehat{\theta})|=
\frac{1}{2\pi} \big|\int_{\Rbb}  e^{-i\eta x} (G(\eta) - \widehat{G}(\eta; \widehat{\theta})) d \eta \big|= \ldots$
\begin{linenomath}
\postdisplaypenalty=0
\begin{align*}
\ldots\!
\le\!
\frac{1}{2\pi}
\int_{\Rbb} \! |\text{Re}_{G}(\eta)-\text{Re}_{\widehat{G}}(\eta; \widehat{\theta})| \!+\! |\text{Im}_{G}(\eta)-\text{Im}_{\widehat{G}}(\eta; \widehat{\theta})|~d\eta
< \frac{1}{2\pi} \big(4 \epsilon_1+ C_1 \epsilon_3 + \frac{C'C_1^2}{{P}}\big).
\end{align*}
\end{linenomath}
As demonstrated above, the derived upper bound for  \(|\min(\widehat{g}(x; \widehat{\theta}), 0)| \) can be decomposed into several error components: truncation error (\(4 \epsilon_1/(2\pi) \)), the MAE regularization term (\(C_1 \epsilon_3/(2\pi) \)),
and  sampling error (\( C'C_1^2/(2\pi P)\)).
\end{remark}

{\siscb{
\begin{remark}[Multi-dimensional]
\label{rm:multi}
We now extend Lemma~\ref{lemma:Gtog} and Remark~\ref{rm:non-neg} to the general \( d \)-dimensional case, where \( d \geq 1 \).
In this context, the target density is \( g(\boldsymbol{x}) \), where \( \boldsymbol{x} \in \mathbb{R}^d \), and \( \widehat{g}(\boldsymbol{x}; \widehat{\theta}) \), with \( \widehat{\theta} \in \widehat{\Theta} \subseteq \Theta \subset \mathbb{R}^{(d+2)N} \), is the estimated density.

 We denote by $\boldsymbol{\eta} = [\eta^{(1)}, \eta^{(2)}, \ldots, \eta^{(d)}]$
the $d$-dimensional Fourier domain vector, restricted to $[-\eta', \eta']^d$,
where $0< \eta'<\infty$ and is sufficiently large. For ease of exposition, we assume there are \( P^{1/d} \) partition points in each dimension, resulting in \( P \) total data points for training. The same (potentially non-uniform) partition is used across all dimensions, with partition intervals satisfying \( \delta_{\min} = C_0 / P^{1/d} \) and \( \delta_{\max} = C_1 / P^{1/d} \), where \( C_0, C_1 > 0 \) are finite constants independent of \( P \) and \( d \). Additionally, we assume
\[
C' := \sup_{\boldsymbol{\eta}, \theta} \left| \nabla_{\boldsymbol{\eta}} \left( |G(\boldsymbol{\eta}) - \widehat{G}(\boldsymbol{\eta}; \theta)|^2 \right) \right| < \infty,
\]
where \( \nabla_{\boldsymbol{\eta}} (\cdot) \) is the gradient with respect to $\boldsymbol{\eta}$.
In this case, the loss function becomes
\EQ
\label{eq:loss_d}
{\text{Loss}}_P(\theta) =
\frac{1}{P} \sum_{p=0}^{P} \l|G(\boldsymbol{\eta}_p) - \widehat{G}(\boldsymbol{\eta}_p; \theta)\r|^2 + R_{P}(\theta),
\EN
where $R_{P}(\theta) = \frac{1}{P} \sum_{p=1}^{P} \left( \left| \text{Re}_{G}(\boldsymbol{\eta}_p) - \text{Re}_{\widehat{G}}(\boldsymbol{\eta}_p; \theta) \right| + \left| \text{Im}_{G}(\boldsymbol{\eta}_p) - \text{Im}_{\widehat{G}}(\boldsymbol{\eta}_p; \theta) \right| \right)$.

Remark~\ref{rm:trunc} also extends to the general multi-dimensional case:
for a given $\epsilon_1>0$, we can find $[-\eta', \eta']^d$ such that
the error from truncating the sampling domain from $\Rbb^d$ to
$[-\eta', \eta']^d$ is bounded by $\epsilon_1$.
Suppose that the parameter $\widehat{\theta}$ learned by training ${\text{Loss}}_{P}(\theta)$, $\theta \in \widehat{\Theta}$, defined in \eqref{eq:loss_d},
is such that ${\text{Loss}}_P(\widehat{\theta}) <\epsilon_2$,
with the regularization term $R_{P}(\widehat{\theta})<\epsilon_3$, where $\epsilon_2, \epsilon_3 >0$.
Then, we have
\EQ
\label{eq:g_g_d}
 \int_{\mathbb{R}^d} \big|g(\boldsymbol{x}) - \widehat{g}(\boldsymbol{x}; \widehat{\theta})\big|^2 ~d\boldsymbol{x}
< \frac{1}{2\pi}
\big(4\epsilon_1 + C_1(\epsilon_2 + \epsilon_3) + \frac{C' C_1^{d +1}}{2P^{1/d}}\big).
\EN
The potential loss of nonnegativity in $\widehat{g}(\cdot;\widehat{\theta})$,
i.e.\ $|\min(\widehat{g}(\boldsymbol{x}; \widehat{\theta}), 0)|$, is bounded by
\EQ
\label{eq:g_loss_d}
|\min(\widehat{g}(\boldsymbol{x}; \widehat{\theta}), 0)| \le \frac{1}{2\pi} \big(4 \epsilon_1+ C_1 \epsilon_3 + \frac{C' C_1^{d +1}}{2P^{1/d}}\big).
\EN
When $d = 1$, we recover the bounds presented in  Lemma~\ref{lemma:Gtog} and Remark~\ref{rm:non-neg}.
The main distinction between the bounds for the one- and multi-dimensional cases arises from the error introduced by the composite left-hand quadrature rule.
\end{remark}
}}

We emphasize that the explicit quantification of bounds for \(L_2\) estimation error and the potential loss of nonnegativity in the estimated transition density, identified and controlled through truncation, training, and sampling errors, highlights FourNet's significant practical advantages. The rigorous analysis of these practical bounds validates our methodology's robustness and ensures its applicability in real-world settings.
{\siscb{In addition, the derived bounds reflect the impact of the dimensionality \(d\), as seen with the term \( \frac{C_1^{d+1}}{P^{1/d}} \) in \eqref{eq:g_g_d}-\eqref{eq:g_loss_d}.}}


The presented analysis offers an in-depth insight into factors influencing the quality of FourNet's approximation for the transition density \(g(\cdot)\).
Crucially, it also draws attention to the significance of the coefficients preceding each error component.
These coefficients act as markers for the worst-case amplification of individual error components, either in the overall $L_2$ estimation error or in potential loss of nonnegativity. Consequently, they serve as signposts, guiding us towards areas where we should focus our efforts for more efficient training and, consequently, reduced error.

Of all components, \(C'\) attracts special attention. This value is directly related to the oscillatory behavior of \(G(\cdot)\), highlighting challenges in approximation. Therefore, curtailing \(C'\) is an important step toward improving FourNet approximation's quality. In the subsequent section, we discuss a straightforward a linear transformation on the input domain which can effectively temper
the oscillatory nature of \(G(\cdot)\), thereby resulting in significantly improved approximation's quality.

\section{Training}
\label{sc:Training}
We now discuss FourNet's data sampling and training algorithms.
\subsection{Linear transformation}
\label{ssc:data} 
{\siscblue{As a starting point for subsequent discussions, we consider a random variable $X$ and denote the characteristic function of the random variable $X$ by $G_X(\eta) = \mathbb{E}[e^{i\eta X}]$, assumed to be available in closed-form. We analyze $\text{Re}_{G_X}(\eta)$ and $\text{Im}_{G_X}(\eta)$ under two scenarios: the Heston model \cite{heston93} (\(X = \ln(S_T)\)) and the Kou model \cite{kou2002jump} (\(X = \ln(S_T/S_0)\)), where \(S_t\) represents the underlying asset price at time \(t \in [0, T]\). Given that $\text{Im}_{G_X}(\eta)$ exhibits similar behaviors across both models, our detailed analysis will focus on $\text{Re}_{G_X}(\eta)$.

Figures~\ref{fig:lt}-(a) and (b) display $\text{Re}_{G_X}(\eta)$ for these cases, with data from Tables~\ref{tab:hes_parameters} and \ref{tab:kou1_parameters}. The Heston model exhibits rapid oscillations within a small domain (about $[-10, 10]$), whereas the Kou model, with a simple oscillation pattern, presents a very large domain (approximately $[-1000, 1000]$).  Both situations pose significant challenges during neural network training: rapid oscillations can lead to numerous local minima and erratic gradients, while large domains can compromise training efficiency.

\begin{figure}[htb!]
      \centering
      \subfigure[untransformed $\text{Re}_{G_X}(\eta)$ - Heston]{
            \includegraphics[width = 0.4\textwidth] {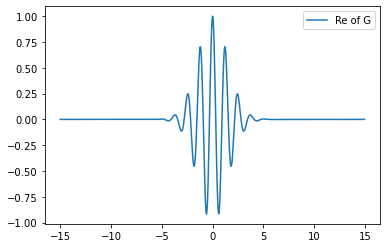}}
  ~
      \subfigure[untransformed $\text{Re}_{G_X}(\eta)$ - Kou]{
            \includegraphics[width = 0.4\textwidth] { 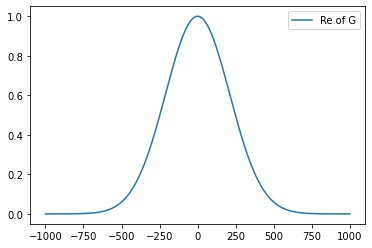}}
        \\
        \subfigure[transformed  $\text{Re}_{G_Y}(\eta)$ - Heston]{
            \includegraphics[width = 0.4\textwidth] {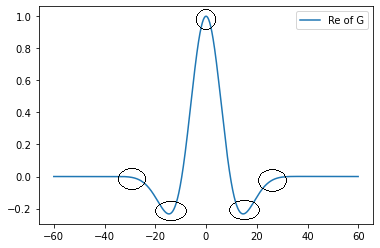}}
        ~
\subfigure[transformed $\text{Re}_{G_Y}(\eta)$ - Kou]{
            \includegraphics[width = 0.4\textwidth] {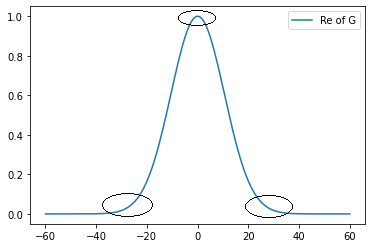}}
      \vspace*{-0.2cm}
      \caption{Comparisons between $G_X(\eta)$ and $G_Y(\eta)$, where $Y = aX + c$; Heston model - (a) and (c): $a=0.15$ and $c= -0.6$; Kou model - (b) and (d): $a=20$ and $c= 0$;
      critical regions of $\text{Re}_{G_Y}$ are highlighted;
      the behaviour of $\text{Im}_{G_X}(\eta)$ is similar (not shown).      }
      \label{fig:lt}
\end{figure}
Motivated by these challenges, we propose a linear transformation \(Y = aX + c\) as a mechanism for adjusting oscillations and controlling domain sizes, resulting in \(G_Y(\eta) = e^{i\eta c}G_X(a\eta)\).
To maintain positive direction and scale in the transformation, $a$ is constrained to be greater than zero.
It modulates oscillation frequency and domain size: \(a < 1\) reduces frequency and expands the domain, while \(a > 1\) compresses it.
The parameter \(c\), a small real number, adjusts the phase of \(G_X(a\eta)\): positive \(c\) shifts the phase forward, negative \(c\) backward. We recommend a \(c\) range from \([-1, 1]\), allowing for significant yet manageable phase shifts across typical \(\eta\) values \cite{oppenheim2017signals}.

Since the overall effectiveness of the transformation heavily depends on the interaction between \(a\) and \(c\) and the properties of \(G_X(\eta)\)--which are highly model-dependent--empirical testing is essential to determine suitable parameter settings.

To illustrate, for the Heston model, \(a = 0.15\) and \(c = -0.6\) reduce oscillation frequency, making $\text{Re}_{G_Y}(\cdot)$ more amenable to NN learning despite a slightly expanded domain $[-60, 60]$ (Figure~\ref{fig:lt}-(c)). In the Kou model, \(a = 20\) (and \(c = 0\)) significantly compresses the domain to $[-50, 50]$, maintaining the same oscillation pattern, which enhances training efficiency (Figure~\ref{fig:lt}-(d)). The critical regions of  $\text{Re}_{G_Y}(\cdot)$ for both models are highlighted, with similar behaviors observed for $\text{Im}_{G_Y}(\cdot)$.

To improve training efficiency further, judicious allocation of sampling data points in crucial areas of both the \(\text{Re}_{G_Y}(\cdot)\) and \(\text{Im}_{G_Y}(\cdot)\) is essential, a strategy to be elaborated in the following subsection.}}
\begin{remark}
\label{rm:linear}
Unless otherwise stated,  throughout our discussion, the characteristic function $G(\cdot)$ employed for the loss function
${\text{Loss}}_{P}(\theta)$ (as defined in \eqref{eq:loss}) corresponds to potentially linearly transformed
characteristic function. Specifically, $G(\cdot) =  G_Y(\cdot)$, where $Y = a X + c$, where $a$ and $c$ are known real constants.
Let $\widehat{g}_Y(y; \widehat{\theta}) = \sum_{n =1}^{N} \widehat{\beta}_n \phi\l(\widehat{w}_{n} y + \widehat{b}_n\r)$, where $\phi = e^{-x^2}$, be {\siscblue{a Fourier-trained FFNN transition density.}} We can recover the estimated transition
density for the random variable $X$ by simply using {\siscblue{$\widehat{g}_X(x; \widehat{\theta}) = \widehat{g}_Y(ax + c; \widehat{\theta}) ~ \left|\frac{d}{dx}(ax + c)\right| = |a|~\widehat{g}_Y(ax + c; \widehat{\theta})$.}}
\end{remark}
\subsection{Sampling data and MAE regularization}
\label{ssc:sampling}
Given our prior knowledge of the (potentially linearly transformed) characteristic function $G(\eta)$ in its closed-form, we strategically concentrate spatial sampling points $\{\eta_p\}_{p = 1}^P$ towards critical regions of $G(\eta)$.
{\siscblue{To identify critical regions, we use symbolic computation to derive the first and second partial derivatives, as well as Hessians for multi-dimensional scenarios, of both the real and imaginary parts of $G(\eta)$.
Critical points and inflection points are determined through these derivatives, with numerical methods applied when closed-form solutions are infeasible.\footnote{For overly complex forms of \(G(\eta)\), we utilize Python to approximate its real/imaginary parts and  visually verify the locations of critical regions.} For visualization, refer to Fig.~\ref{fig:lt} (c) and (d), which highlight critical regions (in circles) for the Heston and Kou models. This methodology allows us to strategically focus our sampling on areas of convexity change, peaks, and other significant features of $\text{Re}_{G}(\eta)$ and $\text{Im}_{G}(\eta)$. Such partitioning of the truncated sampling domain $[-\eta', \eta']$ is achieved through a mapping function, such as the $\texttt{sinh}(\cdot)$-based function, which transforms uniform grids into non-uniform ones with a concentration of points in these critical regions.}} It is noteworthy that similar methodologies for point construction have found successful applications as evidenced in \cite{tavella2000pricing, christara_2010, dang2015efficient}. A partitioning scheme that addresses such scenarios with multiple concentration points is presented in Appendix~\ref{app:non-uni}. We emphasize that a randomly sampled dataset of $\{\eta_p\}_{p = 1}^P$ might inadequately cover these crucial regions, often requiring a significantly larger dataset for the same precision.

With our strategically defined set \(\{\eta_p\}_{p = 1}^P\) in place, we emphasize the role of the MAE regularization \(R_{P}(\theta)\) in our optimization process. It allows the optimization to focus on concentrating efforts to reduce discrepancies specifically in critical regions
while potentially allowing for some discrepancies in less essential areas. Through this, we aim to strike a balance between precision and generalization, thereby curbing potential over-fitting. Our comprehensive numerical tests, presented in Section~\ref{sec:num}, suggest that this combined approach -
strategic sampling based on $\G$ characteristics and employing MAE regularization \eqref{eq:reg} - is efficient and robust.

\subsection{Training considerations}
We briefly describe key considerations in FourNet's training the ${\text{Loss}}_P(\cdot)$ to obtain the empirical minimizer $\widehat{\theta}^*$.
The training of FFNNs is divided into two main stages: the rapid exploration phase and the refinement phase. The initial phase seeks to find a good set of initial weights for the FFNN and fine-tune the baseline learning rate, as these initial weights significantly impact the convergence and accuracy of the training. The refinement phase focuses on further perfecting these weights, often necessitating reduced learning rates to achieve meticulous updates. Due to different focuses of the two stages, choosing the right optimizer for each phase is essential. The Adaptive Moment Estimation (Adam) \cite{adam} and AMSGrad with a Modified Stochastic Gradient \cite{AMSgrad} are standout candidates.

\begin{wrapfigure}{r}{0.475\textwidth}
\begin{center}
\vspace*{-0.2cm}
            \includegraphics[width = 0.45\textwidth] {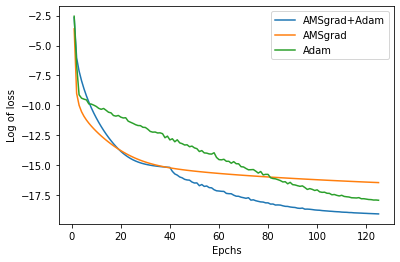}
      \caption{Comparisons among AMSgrad+Adam, Adam, and AMSgrad for the loss function
      ${\text{Loss}}_P(\cdot)$ corresponding to  Figures~\ref{fig:lt} (c) and (d).}
      \label{fig:optimizer}
\end{center}
\vspace*{-0.6cm}
\end{wrapfigure}

Figure~\ref{fig:optimizer} presents a visual comparative analysis of the performance of these optimizers is compared in terms of reducing the empirical loss function ${\text{Loss}}_P(\cdot)$ over a series of training epochs for the case of the Heston model. As illustrated therein, AMSGrad achieves a smoother and steeper reduction in the ${\text{Loss}}_P(\cdot)$ compared to Adam, especially at higher initial learning rates. However, as the epochs progress, Adam tends to surpass AMSGrad. Our proposed methodology suggests employing AMSGrad during the rapid exploration phase and switching to Adam during the refinement phase.

Putting everything together, a single-layer FFNN algorithm for estimating the transition density by learning its Fourier transform is given in Algorithm~\ref{algRong}.
\begin{algorithm}[h]
\caption{Algorithm for approximating the transition density function $g(\cdot)$ using a FFNN trained
in the Fourier domain, given a closed-form expression of the Fourier transform $G(\cdot)$}
\label{algRong}
\begin{algorithmic}[1]
\STATE
using a closed-form expression of $G(\cdot)$ and numerical integration
to find sufficiently large $\eta'$ as per \eqref{eq:Rb2};
\STATE
initialize $N$ (number of neurons), $P$ (number of samples);
generate $\{\eta_p\}_{p = 1}^P$ on $[-\eta', \eta']$ using a non-uniform partitioning algorithm (see Algorithm~\ref{alg:mapping_pro1});
\STATE
use AMSgrad optimizer in first training stage to
find a good set of initial weights and fine-tune the baseline learning rate;

\STATE  use Adam optimizer in the second training stage
\STATE  construct $\widehat{g}(\cdot, \widehat{\theta})$ with $\ds\widehat{\theta} \in \arg\min_{\theta \in \widehat{\Theta}} {\text{Loss}}_{N}(\theta)$,
where ${\text{Loss}}_{N}(\theta)$ is defined in \eqref{eq:loss};
\end{algorithmic}
\end{algorithm}

\noindent The computational complexity of Algorithm~\ref{algRong} is primarily determined by the training of the single-layer FFNN, involving both forward and backward passes.
During the forward pass, the NN computes outputs using operations like matrix-vector multiplications and activation function evaluations, with complexity proportional to $N(d+3)$ (flops), assuming a fixed batch size.
The backward pass, crucial for gradient computation and parameter updates using AMSGrad and Adam optimizers, mirrors this complexity.
Thus, for \(P\) data points and \(H\) epochs in both the forward and backward passes, the total complexity is $\mathcal{O}\big(2PHN(d+3)\big)$ (flops) or $\mathcal{O}(PHNd)$ (flops).

\section{Numerical experiments}
\label{sec:num}
In this section, we demonstrate FourNet's accuracy and versatility through extensive examples.
To measure the accuracy of FourNet, we define several (empirical) metrics. Specifically, the closeness of two elements $f_1$ and $f_2$ of $L_p\left(\mathbb{R}\r)$, $p\in \{1, 2\}$, is measured by $L_p(f_1, f_1) = \int_{[-A, A]}  |(f_1(x)- f_2(x))|^p~ dx$, for $ A>0$ sufficiently large.
In addition, the Maximum Pointwise Error (MPE) is defined by $\text{MPE} (f_1, f_2) = \max_{1\le k\le K} |f_1(x_k)-f_2(x_h)|$,
where  $\{x_k\}_{k = 1}^K$ is the set of evaluation points. Among these, $L_2$-error stands out as the principal metric,
underscored by the $L_2$ error analysis presented in Section~\ref{sc:error}.

In our experiments, unless otherwise stated, all integrals, including those appear in pricing an option, are computed using adaptive Gauss quadrature rule (based on QUADPACK library in Fortran 77 library, \texttt{quad} function in Python).

\subsection{Setup and preliminary observations}
\label{ssc:sample}

Informed by Remarks~\ref{rm:trunc} and \ref{rm:non-neg}, for all numerical experiments carried out in this paper,
the sampling domain $[-\eta', \eta']$ (in the Fourier space), the number of samples $P$ are chosen sufficiently large.
Specifically, in computing a sufficiently large $\eta'$, given a closed-form expression for $G(\cdot)$,
we perform numerical integration to estimate $\eta'$ such that \eqref{eq:Rb2} corresponding to $G(\cdot)$
is satisfied for a tolerance $\epsilon_1 = 10^{-7}$.
That is, with $D = \Rbb\setminus[-\eta',\eta']$, we have
\EQ
\int_{D} |\text{Re}_{G}(\eta)|~d\eta < \epsilon_1,\quad
\int_{D} |\text{Im}_{G}(\eta)|~d\eta < \epsilon_1, \quad
\int_{D} |G(\eta)|^2~d\eta < \epsilon_1.
\label{eq:F}
\EN
This typically results in $[-\eta', \eta'] = [-60, 60]$ for all models considered hereafter.

\begin{center}
\begin{tabular}{|p{0.14\textwidth}|>{\raggedleft\arraybackslash}p{0.1\textwidth}|>{\raggedleft\arraybackslash}p{0.1\textwidth}|}
\hline
 & 1-D& 2-D \\
 & ($d = 1$)& ($d = 2$)\\
\hline
$N$& 45& 45   \\
\hline
$P$& $10^{6}$&  $10^{6}$  \\
\hline
\# $epochs_1$& 5& 6 \\
\hline
\#  $epochs_2$ & 100& 40\\
\hline
$l_1$& 0.0015& 0.04 \\
\hline
$l_2$& 0.0012&0.00025 \\
\hline
$batch size$& 1024& 1024 \\
\hline
time (mins)& 3.2& 22.5 \\
\hline
\end{tabular}
\captionof{table}{{\siscblue{Hyperparameters for NN training  with typical training times.}}}
\label{tab:hyper}
\end{center}

The training setup utilizes a system equipped with an Intel Core i7-13705H Processor,
operating on Windows 11 with 32GB of memory and 1TB of storage. The environment runs Python 3.10 and TensorFlow 2.8. Detailed hyperparameters for all experiments are outlined in Table~\ref{tab:hyper}, which specifies the number of epochs and learning rates for both the rapid exploration phase (\# $epochs_1$ and $l_1$) and the refinement phase (\# $epochs_2$ and $l_2$). All datasets feature one-dimensional inputs ($d = 1$), except for the two-dimensional Merton jump-diffusion process discussed in Section~\ref{ssc:merton2}. The table also presents typical total training times for both phases, aggregated over 30 runs.

The parameters $\widehat{\theta}$ is learned through training, satisfying ${\text{Loss}}_P(\widehat{\theta}) \le 10^{-6}$. This implies that the MAE regularization term
$R_{P}(\theta)$, as defined in \eqref{eq:reg}, is less than $10^{-6}$.
{\myblue{ We observe that the measure for loss of non-negativity $|\min(\widehat{g}(\cdot; \widehat{\theta}), 0)|$ is about  $10^{-7}$, negligible for all practical purposes.
For comparison, we evaluate the bound
$\varepsilon = \frac{1}{2\pi} \big(4 \epsilon_1+ C_1 \epsilon_3 + \frac{C'C_1^2}{{P}}\big)$, as presented in Remark~\ref{rm:non-neg}.
We take $\epsilon_1 = 10^{-7}$ (in accordance with \eqref{eq:F}), and
$\epsilon_3 = 10^{-6}$, given that ${\text{Loss}}_P(\widehat{\theta}) \le 10^{-6}$).
Considering a uniform partition, we have \(\delta_p = 120/P\) for all \(p\). Consequently, \(C_1 \ge 120\). Using a conservative estimate, we take \(C_1 = 120\), yielding \(C_1 \epsilon_3 \approx 10^{-4}\) and \(\frac{C' C_1^2}{P} \approx C'10^{-2}\).
Our numerical findings suggest a notable reduction in the loss of non-negativity when the linear transformation highlighted in Subsection~\ref{ssc:data} is used. This transformation diminishes \(C'\), suggesting that \(\frac{C' C_1^2}{P} \approx C'10^{-2}\) is the primary contributing term.}}

We now explore FourNet's accuracy in estimating transition densities
a broad array of  dynamics commonly encountered in quantitative finance.
Subsequently, we will focus on its application for pricing both European and Bermudan options.

\subsection{Transition densities}
\subsubsection{Exponential L\'{e}vy processes}
We select models that are well-known within the domain of exponential L\'{e}vy processes, where the L\'{e}vy-Khintchine formula provides a clear representation for the characteristic function $G(\cdot)$ as detailed in \cite{ken1999levy}. As example, we focus on the Merton jump-diffusion model, introduced by \cite{merton1976option}, and the CGMY model as proposed by \cite{carr2002fine}. It's worth noting that the CGMY model can be seen as an extension of the Variance-Gamma model, originally presented in \cite{madan1990variance}. Additionally, while we conducted tests on the Variance-Gamma model and the Kou jump-diffusion model \cite{kou2002jump}, FourNet consistently proved to be very accurate. In fact, the outcomes from these tests align so closely with those of the highlighted models that we have chosen not to detail them {\siscblue{in this subsection}}
for the sake of brevity.

In exponential L\'{e}vy processes, with $\{S_t\}_{t = 0}^T$ being the price process, the process $\{X_t\}_{t = 0}^T$,
where $X_t = \ln\l(S_t/S_0\r)$,  is a L\'{e}vy  process. Relevant to our discussions is the fact that the
characteristic function of the random variable $X_t$ is $G_X(\eta) = \exp(t \psi(\eta))$ \cite{ken1999levy}.
As in all numerical examples on transition densities presented in this section, we take $t = T$ which is specified below. The characteristic exponent $\psi(\eta)$ for various exponential L\'{e}vy processes considered in this paper are given subsequently.

\vspace*{+0.1cm}
\noindent\textbf{Merton jump-diffusion dynamics \cite{merton1976option}} In this case, the characteristic exponent $\psi(\eta)$
is given by
$\psi(\eta) = i\left(\mu-\frac{\sigma^2}{2}\right) \eta-\frac{\sigma^2 \eta^2}{2}+\lambda\left(e^{i \tilde{\mu} \eta-\tilde{\sigma}^2 \eta^2 / 2}-1\right)$.
In this case, a semi-explicit formula for $g(x;T)$ is given by (see \cite{Zhang2023}[Corollary 3.1])
\EQ
\label{eq:merton_g}
g(x;T) = \sum_{k=0}^{\infty}  \frac{e^{-\lambda T} (\lambda T)^k}{k!} g_{\norm}\big(x; \big(\mu-\frac{\sigma^2}{2}-\lambda \kappa\big)T + k\tilde{\mu} , \sigma^2 T + k\tilde{\sigma}^2\big).
\vspace*{-0.25cm}
\EN
Here, $\kappa = e^{\tilde{\mu} + \tilde{\sigma}^2/2} - 1$, and $g_{\norm}(x; \mu', (\sigma')^2)$ denotes the probability density function of a normal random variable  with mean $\mu'$ and
variance $(\sigma')^2$. The semi-explicit formula given by \eqref{eq:merton_g} serves as our reference density against which we validate the estimated transition density produced by FourNet. Computationally, we truncate the infinite series in \eqref{eq:merton_g} to 15 terms. The approximation error resulting from this truncation is approximately $10^{-20}$, which is sufficiently small for all practical intents and purposes.
\\
\begin{minipage}[t]{0.475\textwidth}
\strut\vspace*{-\baselineskip}\newline
\flushleft
\begin{tabular}{|p{0.6\textwidth}>{\raggedleft\arraybackslash}p{0.2\textwidth}|}
\hline
Parameters              & Values \\
\hline
$T$ (maturity in years)                           &1\\
$S_0$ (initial asset price)                          &100\\

$r$ (risk free rate)                        & 0.05  \\
$\sigma$ (volatility)                   & 0.15  \\
$\lambda$ (jump intensity)                  & 0.1  \\
$\tilde{\mu}$ (mean of jump size)               & -1.08 \\
$\tilde{\sigma}$ (std of jump size)           & 0.4  \\
\hline
\end{tabular}
\captionof{table}{Parameters for the Merton jump diffusion dynamics;
values are taken from \cite{ForsythLabahn2017}[Table 4].}
\label{tab:MJD_parameters}
\end{minipage}
~~
\begin{minipage}[t]{0.4755\textwidth}
\strut\vspace*{-\baselineskip}\newline
\flushright
\begin{tabular}{|p{0.5\textwidth}>{\raggedleft\arraybackslash}p{0.35\textwidth}|}
\hline
$N$ ($\#$ of neurons)   & 45
\\
\hline
$L_1\l(\text{Re}_{G}, \text{Re}_{\widehat{G}}\r)$       &$2.4\times10^{-04}$
\\
\rowcolor{gray!40}
$L_2\l(\text{Re}_{G}, \text{Re}_{\widehat{G}}\r)$        &$1.4\times10^{-09}$
\\
$\text{MPE}\l(\text{Re}_{G}, \text{Re}_{\widehat{G}}\r)$  & $1.0\times10^{-05}$
\\
\hline
$L_1\l(\text{Im}_{G}, \text{Im}_{\widehat{G}}\r)$ & $5.8\times10^{-04}$
\\
\rowcolor{gray!40}
$L_2\l(\text{Im}_{G}, \text{Im}_{\widehat{G}}\r)$ & $1.8\times10^{-09}$
\\
$\text{MPE}\l(\text{Im}_{G}, \text{Im}_{\widehat{G}}\r)$ & $2.4\times10^{-05}$
\\
\hline
\rowcolor{gray!40}
$L_2\l(g, \widehat{g}\r)$& $1.6\times10^{-09}$
\\
\hline
\end{tabular}
\captionof{table}{Estimation errors for the Merton model;
parameters from  Table~\ref{tab:MJD_parameters};
linear transform in Remark~\ref{rm:linear} used with $(a,c)=(0.6, 0.08)$.
}
\label{tab:MJD_error_g}
\end{minipage}

\noindent  The parameters used for this test case are given in Table~\eqref{tab:MJD_parameters}.
The linear transform in Remark~\ref{rm:linear} is used with $(a,c)=(0.6, 0.08)$.
The number of neurons ($N$) and $L_p$/MPE estimation errors by FourNet are presented in Table~\ref{tab:MJD_error_g},
with the principal metric $L_2$-error highlighted.
As evident, FourNet is very accurate with negligible $L_2$ estimation error (of order $10^{-9}$).
We note that, without a linear transform, the resulting $L_p$/MPE estimation errors are much larger. For example,
$L_2\l(\text{Re}_{G}, \text{Re}_{\widehat{G}}\r)=3.3\times10^{-7}$ instead of $1.4\times10^{-9}$
and $L_2\l(\text{Im}_{G}, \text{Im}_{\widehat{G}}\r) =1.8\times 10^{-7}$ vs $1.8\times 10^{-9}$.

\noindent\textbf{CGMY model \cite{carr2002fine}}
In this case, the characteristic exponent $\psi(\eta)$
is given by $\psi(\eta) = CG\l(\eta \r) + i\eta\l(r+\varpi\r)$, where $CG(\eta)=C \Gamma(-Y)\big[(M-i \eta)^Y-M^Y + (G+i \eta)^Y-G^Y\big]$
and  $\varpi = -CG(-i)$. Here, $\Gamma(\cdot)$ represents the gamma function. In the CGMY model,
the parameter should satisfy $C \ge 0$, $G \ge 0$, $M \ge 0$ and $Y < 2$.

The parameters used for this test case are given in Table~\eqref{tab:CGMY_parameters}.
The linear transform in Remark~\ref{rm:linear} is used with $(a,c)=(0.5, 0.0)$.
The number of neurons ($N$) and $L_p$/MPE estimation errors by FourNet are presented in Table~\ref{tab:CGMY_error_g},
with the principal metric $L_2$-error highlighted.
Again, it is clear that  FourNet is very accurate, with negligible $L_2$ estimation error (of order $10^{-8}$).
\\
\begin{minipage}[t]{0.475\textwidth}
\strut\vspace*{-\baselineskip}\newline
\flushleft
\begin{tabular}{|p{0.7\textwidth}>{\raggedleft\arraybackslash}p{0.15\textwidth}|}
\hline
Parameters              & Values \\
\hline
$T$ (maturity in years)                           &1 \\
$S_0$ (initial asset price)                          &100\\

$r$ (risk free rate)                        & 0.1 \\
$C$ (overall activity)                & 1 \\
$G$ (exp.\ decay on right)             & 5  \\
$M$ (exp.\ decay on left)            & 5 \\
$Y$ (finite/infinite activity)          & 0.5  \\
\hline
\end{tabular}
\captionof{table}{Parameters for the CGMY dynamics;
values taken from \cite{Fang2008}[Equation 56].} \label{tab:CGMY_parameters}
\end{minipage}
~~
\begin{minipage}[t]{0.475\textwidth}
\strut\vspace*{-\baselineskip}\newline
\flushright
\begin{tabular}{|p{0.5\textwidth}>{\raggedleft\arraybackslash}p{0.3\textwidth}|}
\hline
$N$ ($\#$ of neurons)   & 45
\\
\hline
$L_1\l(\text{Re}_{G}, \text{Re}_{\widehat{G}}\r)$                           &$7.8\times10^{-4}$
\\
\rowcolor{gray!40}
$L_2\l(\text{Re}_{G}, \text{Re}_{\widehat{G}}\r)$                           &$1.7\times10^{-8}$
\\
$\text{MPE}\l(\text{Re}_{G}, \text{Re}_{\widehat{G}}\r)$  & $3.1\times10^{-5}$
\\
\hline
$L_1\l(\text{Im}_{G}, \text{Im}_{\widehat{G}}\r)$ & $2.2\times10^{-4}$
\\
\rowcolor{gray!40}
$L_2\l(\text{Im}_{G}, \text{Im}_{\widehat{G}}\r)$ & $1.3\times10^{-9}$
\\
$\text{MPE}\l(\text{Im}_{G}, \text{Im}_{\widehat{G}}\r)$ & $9.2\times10^{-6}$
\\
\hline
\end{tabular}
\captionof{table}{
Estimation errors for the CGMY model;
parameters from  Table~\ref{tab:CGMY_parameters};
linear transform in Remark~\ref{rm:linear} is employed with $(a,c)=(0.5, 0.0)$.}
\label{tab:CGMY_error_g}
\vfill
\vspace*{-1.5cm}
\end{minipage}

\subsubsection{Heston and Heston Queue-Hawkes}
Moving beyond exponential {\siscblue{L\'{e}vy}} processes, we first evaluate the applicability of FourNet to the Heston model \cite{heston93},
followed by an investigation of the Heston Queue-Hawkes model, as presented in \cite{arias2022}.
For these models, the characteristic functions of the log-asset price, $\ln(S_t)$, over $t \in [0, T]$, are available in closed-form. Despite the non-homogeneous variance features, our focus here is on estimating the transition density of the process $\{\ln(S_t)\}_{t \in [0, T]}$ for European option pricing. Since these options do not require time-stepping for valuation, we can efficiently estimate transition densities using a single FFNN training session.

\noindent\textbf{Heston model \cite{heston93}} The log-price $\ln(S_t)$ and its variance $V_t$ follow the dynamics
\begin{align*}
d\ln(S_t) = \big(r-\frac{V_t}{2}\big) dt + \sqrt{V_t}\, dW_t^{(1)}, \quad
dV_t = \kappa (\bar{V} - V_t)\, dt + \sigma\,\sqrt{V_t}\, dW_t^{ (2)},
\end{align*}
with $S_0 >0$ and $V_0 >0$ given. Here, $\kappa, \bar{V}>0$, and $\sigma > 0$ are constants representing the mean-reversion rate,
the long-term mean level of the variance, and the instantaneous volatility of the variance;
$\{W_t^{ (1)}\}$ and $\{W_t^{ (2)}\}$ are assumed to be correlated with correlation
coefficient $\rho \in [-1, 1]$. As presented in \cite{rollin2009new}[Equation~5], the characteristic function of $X_T = \ln(S_T)$  is given by
\begin{align}
\label{eq:hestonG}
 G^{\heston}_X(\eta)&= \exp^{i r\eta T} \exp \left\{i \ln(S_0)\eta\right\}\left(\frac{e^{\kappa T / 2}}{\cosh (d T / 2)+\xi \sinh (d T / 2) / d}\right)^{2 \kappa \bar{V} / \sigma^2}
 \nonumber
 \\
 &\qquad \qquad  \quad  \cdot \exp \left\{-V_0 \frac{\left(i \eta+\eta^2\right) \sinh (d T / 2) / d}{\cosh (d T / 2)+\xi \sinh (d T / 2) / d}\right\},
\end{align}
where $d=d(\eta)=\sqrt{(\kappa-\sigma \rho i \eta)^2+\sigma^2\left(i \eta+\eta^2\right)}$ and $\xi=\xi(\eta)=\kappa-\sigma \rho \eta i$.

Numerical experiments for the Heston model utilize the parameters listed in Table~\ref{tab:hes_parameters} (the Feller condition is met). Estimation errors are documented in Table~\ref{tab:hes_error_g}, noting the linear transformation Remark~\ref{rm:linear}. In Figure~\ref{fig:hesMER_MIM}, we display plots of the benchmark real part in (a) and the imaginary part in (b). Corresponding estimations by FourNet are also showcased, with the transition density estimated by FourNet depicted in (c). We emphasize FourNet's outstanding performance, particularly evident in the minimal $L_2$ estimation error.
\\
\begin{minipage}[t]{0.475\textwidth}
\strut\vspace*{-\baselineskip}\newline
\flushleft
\begin{tabular}{|p{0.7\textwidth}>{\raggedleft\arraybackslash}p{0.15\textwidth}|}
\hline
Parameters              & Values \\
\hline
$T$ (maturity in years)                           &5\\
$S_0$ (initial asset price)                          &100\\

$r$ (risk free rate)                        & 0.15  \\
$\sigma$ (volatility of volatility)                   & 0.3  \\

$\kappa$ (mean-reversion rate)                  & 3  \\
$\bar{V}$ (mean of volatility)               &0.09 \\
$V_0$ (initial volatility )           & 0.2  \\
$\rho$ (correlation)  &0.4\\
\hline
\end{tabular}
\captionof{table}{Parameters for the Heston model;
values taken from \cite{dang2018dimension}[Table 1].} \label{tab:hes_parameters}
\end{minipage}
~~
\begin{minipage}[t]{0.45\textwidth}
\strut\vspace*{-\baselineskip}\newline
\flushright
\begin{tabular}{|p{0.45\textwidth}>{\raggedleft\arraybackslash}p{0.35\textwidth}|}
\hline
$N$ ($\#$ of neurons)  & 45
\\
\hline
$L_1\l(\text{Re}_{G}, \text{Re}_{\widehat{G}}\r)$                           &$1.94\times10^{-05}$
\\
\rowcolor{gray!40}
$L_2\l(\text{Re}_{G}, \text{Re}_{\widehat{G}}\r)$                           &$6.94\times10^{-12}$
\\
$\text{MPE}\l(\text{Re}_{G}, \text{Re}_{\widehat{G}}\r)$  & $1.18\times10^{-06}$
\\
\hline
$L_1\l(\text{Im}_{G}, \text{Im}_{\widehat{G}}\r)$ & $3.91\times10^{-05}$
\\
\rowcolor{gray!40}
$L_2\l(\text{Im}_{G}, \text{Im}_{\widehat{G}}\r)$ & $5.22\times10^{-11}$
\\
$\text{MPE}\l(\text{Im}_{G}, \text{Im}_{\widehat{G}}\r)$ & $3.74\times10^{-06}$
\\
\hline
\end{tabular}
\captionof{table}{Estimation errors for the Heston model;
parameters from  Table~\ref{tab:hes_parameters};
linear transform in Remark~\ref{rm:linear} is employed with $(a,c)=(0.15, -0.6)$.
}
\label{tab:hes_error_g}
\vfill
\end{minipage}

\begin{figure}[htb!]
	\centering
	\subfigure[$\text{Re}_{G}$, $\text{Re}_{\widehat{G}}$]{
		\label{fig:hesre}
		\includegraphics[width = 0.31\textwidth] {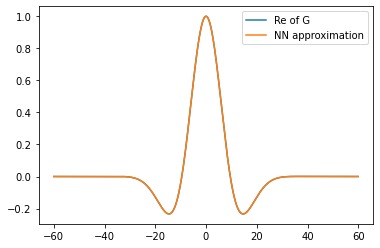}}
    	\subfigure[$\text{Im}_{G}$, $\text{Im}_{\widehat{G}}$]{
		\label{fig:hesim}
		\includegraphics[width = 0.31\textwidth] { 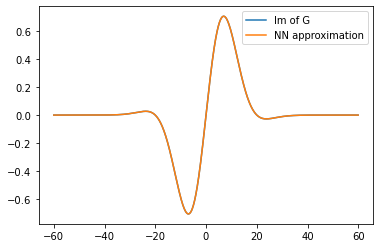}}
        \subfigure[$\widehat{g}(\cdot; T)$]{
		\label{fig:hesdensity}
		\includegraphics[width = 0.31\textwidth] {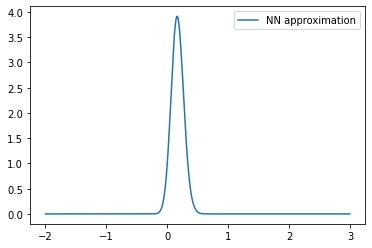}}
	\caption{Heston model corresponding to Table~\ref{tab:hes_parameters} and Table~\ref{tab:hes_error_g}.}
\label{fig:hesMER_MIM}
\end{figure}

\vspace*{+0.25cm}
\noindent\textbf{Heston Queue-Hawkes \cite{daw2022ephemerally, arias2022}}
With $t \in [0, T]$, let $t^{\pm} = \lim_{\epsilon \searrow 0} (t \pm \epsilon)$. Informally, $t^-$ ($t^+$)  denotes the instant of time immediately before (after) calendar time $t$.
The risk-neutral Heston Queue-Hawkes dynamics of the stock price are given by \cite{arias2022}:
\begin{align*}
 &d \l(\frac{S_t}{S_{t^-}}\r) = \l(r-\mu_Y \lambda_{t^-}\r) dt + \sqrt{V_t} dW_t^{(1)} + \l(\exp(Y_t ) - 1\r) d\pi_t, \\
 &dV_t = \kappa (\bar{V}- V_t) dt+ \sigma \sqrt{V_t}  d W_t^{(2)}.
\end{align*}
Here, $\{V_t\}$ is the variance process; $\{W_t^{ (1)}\}$, $\{W_t^{ (2)}\}$ are correlated standard Brownian motions with the constant correlation $\rho \in [-1, 1]$;
$Y_t \sim \text{Normal}(\mu_Y, \sigma_Y^2)$; $\kappa>0$, $\bar{V}>0$ and $\sigma>0$ are the variance's speed of mean reversion, long-term mean, and volatility of volatility parameters, respectively.  Finally, $\{\pi_t\}$ is a counting process with stochastic intensity $\lambda_t$ satisfying the Queue-Hawkes process:
$d \lambda_t = \alpha(d\pi_t - d\pi_t^Q)$, where $\pi_t^Q$ is a counting process with intensity $\beta Q_t$, the constants $\alpha$, and  $\beta$ respectively are the clustering and expiration rates.

The characteristic function of $X_T = \ln(S_T)$ is given by \cite{arias2022}[Equation~6]:
\begin{align}
\label{eq:HQH_den}
G_X^{\HQH}(\eta) = G_X^{\heston}(\eta)~G_{M}(\eta).
\end{align}
Here, $G_X^{\heston}(\eta)$ is given in \eqref{eq:hestonG}, and $G_{M}(\eta)$ is defined as follows
\begin{align*}
G_{M}(\eta)&=  e^{\frac{\lambda^* T}{2 \alpha}\left(\beta-\alpha-i \alpha {\mu}_Y \eta-f(\eta)\right)}
 \cdot\bigg(\frac{2 f(\eta)}{f(\eta)+g(\eta)+e^{-T f(\eta)}(f(\eta)-g(\eta))}\bigg)^{\frac{\lambda^*}{\alpha}} \\
& \cdot\bigg(\frac{\left(1-e^{-T f(\eta)}\right)\left(2 \beta \right)}{f(\eta)+g(\eta)+e^{-T f(\eta)}(f(\eta)-g(\eta))}\bigg)^{Q_0}.
\end{align*}
Here, $f(\eta)=\sqrt{\left(\beta+\alpha\left(1+i \eta {\mu}_Y\right)\right)^2-4 \alpha \beta \psi_Y(\eta)}$, $g(\eta)=\beta+\alpha\left(1+i \eta {\mu}_Y\right)$, $\psi_Y$ is  the characteristic function of normal random variable with mean ${\mu}_Y$ and std ${\sigma}_Y$.

We conduct numerical experiments using the parameters listed in Table~\ref{tab:HQH_parameters}, with estimation errors detailed in Table~\ref{tab:HQH_error_g}, noting the linear transformation Remark~\ref{rm:linear}. Figure~\ref{fig:hesMER_MIM}, we present several plots for the benchmark real/imaginary part in (a)/(b) and  respective results obtained by FourNet, as well as the estimated transition density in (c). Impressively, FourNet demonstrates outstanding $L_2$ estimation accuracy.
\\
\begin{minipage}[t]{0.475\textwidth}
\strut\vspace*{-\baselineskip}\newline
\flushleft
\begin{tabular}{|p{0.7\textwidth}>{\raggedleft\arraybackslash}p{0.15\textwidth}|}
\hline
Parameters              & Values \\
\hline
$T$ (maturity in years)                           &1 \\
$S_0$ (initial asset price)                          &9\\
$V_0$ (initial volatility)                          &0.0625\\
$\bar{V}$ (mean of volatility)  & 0.16\\
$r$ (risk free rate)                        & 0.1  \\
$\sigma$ (volatility of volatility)                   & 0.9  \\
$Q_0$ (initial value of $Q_t$)  &2\\
$\alpha$  (clustering rate) & 2\\
$\beta$ (expiration rate) &3\\
$\lambda^*$ (baseline jump intensity)                  & 1.1  \\
$\mu_Y$ (mean of jump size)               & -0.3 \\
$\sigma_Y$ (std of jump size)           & 0.4  \\

$\rho$ (correlation)  & 0.1\\
\hline
\end{tabular}
\captionof{table}{Parameters for the Heston Queue-Hawkes model.
values are taken from \cite{arias2022}[Table 1].} \label{tab:HQH_parameters}
\end{minipage}
~~~
\begin{minipage}[t]{0.45\textwidth}
\strut\vspace*{-\baselineskip}\newline
\flushright
\begin{tabular}{|p{0.45\textwidth}>{\raggedleft\arraybackslash}p{0.35\textwidth}|}
\hline
$N$ ($\#$ of neurons)  & 45
\\
\hline
$L_1\l(\text{Re}_{G}, \text{Re}_{\widehat{G}}\r)$                           & $3.36\times10^{-4}$
\\
\rowcolor{gray!40}
$L_2\l(\text{Re}_{G}, \text{Re}_{\widehat{G}}\r)$                           &$4.19\times10^{-9}$
\\
$\text{MPE}\l(\text{Re}_{G}, \text{Re}_{\widehat{G}}\r)$  & $2.44\times10^{-5}$
\\
\hline
$L_1\l(\text{Im}_{G}, \text{Im}_{\widehat{G}}\r)$ & $3.38\times10^{-4}$
\\
\rowcolor{gray!40}
$L_2\l(\text{Im}_{G}, \text{Im}_{\widehat{G}}\r)$ & $4.66\times10^{-9}$
\\
$\text{MPE}\l(\text{Im}_{G}, \text{Im}_{\widehat{G}}\r)$ & $1.87\times10^{-5}$
\\
\hline
\end{tabular}
\captionof{table}{
Estimation errors for the Heston Queue-Hawkes model;
parameters from  Table~\ref{tab:HQH_parameters};
linear transform in Remark~\ref{rm:linear} with $(a,c)=(0.18, -0.31)$ is applied to
$G_X^{\HQH}(\cdot)$ in \eqref{eq:HQH_den}.}
\label{tab:HQH_error_g}
\vfill
\end{minipage}
\begin{figure}[htb!]
\vspace*{-0.25cm}
	\centering
	\subfigure[$\text{Re}_{G}$, $\text{Re}_{\widehat{G}}$]{
		\label{fig:HQHre}
		\includegraphics[width = 0.31\textwidth] {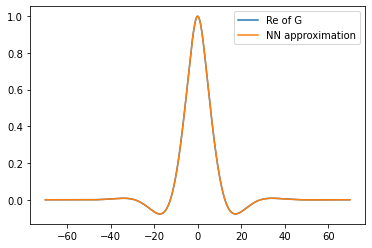}}
    	\subfigure[$\text{Im}_{G}$, $\text{Im}_{\widehat{G}}$]{
		\label{fig:HQHim}
		\includegraphics[width = 0.31\textwidth] { 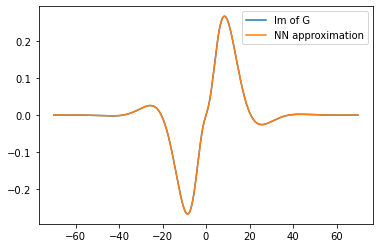}}
        \subfigure[$\widehat{g}(\cdot; T)$]{
		\label{fig:HQHdensity}
		\includegraphics[width = 0.31\textwidth] {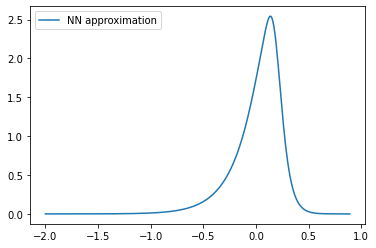}}
	\caption{Heston Queue-Hawkes model corresponding to Table~\ref{tab:HQH_parameters} and Table~\ref{tab:HQH_error_g}.}
\label{fig:HQH_MIM}
\end{figure}

\subsection{Two-dimensional Merton jump-diffusion process}
\label{ssc:merton2}
We now demonstrate the capability of FourNet to two-dimensional Merton jump-diffusion process \cite{ruijter2012two}.
The stock prices follow the risk-neutral dynamics
\EQ
dS_t^{(\ell)} = (r-\lambda\kappa^{(\ell)})S_t^{(\ell)}dt + \sigma^{(\ell)} S_t^{(\ell)} dW_t^{(\ell)} + \big(e^{Y^{(\ell)}} - 1\big)S_t^{(\ell)} d\Pcal_t,
~\ell = 1, 2,
\EN
with $S_0^{(\ell)} >0$ given. Here, $r>0$ is risk free rate and $\sigma^{(\ell)}>0$, $\ell = 1, 2$, are instantaneous volatility for the $\ell$-underlying;
$\{W_t^{(1)}\}$ and $\{W_t^{(2)}\}$ are  standard  Brown motions with correlation $\rho \in [-1, 1]$; $\{\Pcal_t\}$ is a Poisson process with a constant finite jump arrival rate $\lambda\ge0$; $[Y^{(1)}, Y^{(2)}]$ are bivariate normally distributed jump sizes
with mean $\boldsymbol{\tilde{\mu}} = [\tilde{\mu}^{(1)}, \tilde{\mu}^{(2)}]$ and covariance matrix for the jump components,
denoted by $\boldsymbol{\widetilde{\Sigma}}$, where $\boldsymbol{\widetilde{\Sigma}}^{(\ell, k)}=\tilde{\sigma}^{(\ell)} \tilde{\sigma}^{(k)}
\tilde{\rho}^{(\ell, k)}$, $\ell, k \in  \{1, 2\}$, with $\tilde{\rho}^{(1, 2)} = \tilde{\rho}^{(2, 1)} = \tilde{\rho}  \in [-1, 1]$;
$\kappa^{(\ell)} = \mathbb{E}[e^{Y_{(\ell)}}-1]$, $\ell = 1, 2$.

For subsequent use, we define
$\boldsymbol{\mu} = [\mu^{(1)}, \mu^{(2)}]$, where $\mu^{(\ell)} = (r-\lambda\kappa^{(\ell)}-(\sigma^{(\ell)})^2/2) T$, $\ell \in  \{1, 2\}$,
and covariance matrix for the diffusion components,
denoted by $\boldsymbol{\Sigma}$, where $\boldsymbol{\Sigma}^{(\ell, k)}=\sigma^{(\ell)} \sigma^{(k)}
\rho^{(\ell, k)}$, $\ell, k \in  \{1, 2\}$, with $\rho^{(1, 2)} = \rho^{(2, 1)} = \rho  \in [-1, 1]$.
\\
\begin{minipage}[t]{0.475\textwidth}
\strut\vspace*{-\baselineskip}\newline
\flushleft
\begin{tabular}{|p{0.7\textwidth}>{\raggedleft\arraybackslash}p{0.15\textwidth}|}
\hline
Parameters              & Values \\
\hline
$T$ (maturity in years)                           &1\\
$\sigma_1$ (volatility)                          &0.12\\
$\sigma_2$                           &0.15\\
$r$ (risk free rate)                        & 0.05  \\
$K$ (strike price)                        & 100  \\
$\rho$ (correlation)                   & 0.3  \\
$\tilde{\rho}$ (jump correlation)                   & -0.2  \\
$\lambda$ (jump intensity)                  & 0.6  \\
$\tilde{\mu}_1$ (jump size mean)               & -0.1\\
$\tilde{\mu}_2$                & 0.1 \\
$\tilde{\sigma}_1$ (jump size std)           & 0.17 \\
$\tilde{\sigma}_2$           &0.13  \\
\hline
\end{tabular}
\captionof{table}{Parameters for the  2D Merton jump diffusion.
Values are taken from \cite{ruijter2012two} [Parameter sets 2].} \label{tab:MJD2_parameters}
\end{minipage}
\begin{minipage}[t]{0.45\textwidth}
\strut\vspace*{-\baselineskip}\newline
\flushright
\begin{tabular}{|p{0.45\textwidth}>{\raggedleft\arraybackslash}p{0.35\textwidth}|}
\hline
$N$ ($\#$ of neurons)  & 45
\\
\hline

\rowcolor{gray!40}
$L_2\l(\text{Re}_{G}, \text{Re}_{\widehat{G}}\r)$                           &$3.15\times10^{-8}$
\\
$\text{MPE}\l(\text{Re}_{G}, \text{Re}_{\widehat{G}}\r)$  & $8.33\times10^{-5}$
\\
\hline

\rowcolor{gray!40}
$L_2\l(\text{Im}_{G}, \text{Im}_{\widehat{G}}\r)$ & $2.24\times10^{-8}$
\\
$\text{MPE}\l(\text{Im}_{G}, \text{Im}_{\widehat{G}}\r)$ & $4.62\times10^{-5}$
\\
\hline
\end{tabular}
\captionof{table}{Estimation errors for the 2D Merton jump-diffusion model;
parameters from  Table~\ref{tab:MJD2_parameters};}
\label{tab:MJD2_error_g}
\vfill
\end{minipage}

\noindent The characteristic function of the random variable $\boldsymbol{X}_T=\big[\ln\big(S_T^{(\ell)}/S_0^{(\ell)}\big)\big]$, $\ell = 1, 2$, is given by \cite{ruijter2012two}[Eqn~(6.7)]
\EQ
G_{\boldsymbol{X}}(\boldsymbol{\eta}) = \exp\big(i\boldsymbol{\mu}' \boldsymbol{\eta} - \frac{1}{2} \boldsymbol{\eta}' \boldsymbol{\Sigma} \boldsymbol{\eta}\big) \exp\bigg(\lambda T\bigg(\exp\big(i\boldsymbol{\tilde{\mu}}' \boldsymbol{\eta}- \frac{1}{2}\boldsymbol{\eta}'\boldsymbol{\widetilde{\Sigma}}\boldsymbol{\eta}\big)-1\bigg)\bigg).
\label{eq:G_X_2D}
\EN
In this case, it is convenient to write $\widehat{g}(\boldsymbol{x};\theta)$ in the following form:
\EQS
\widehat{g}(\boldsymbol{x};\theta) =
\sum_{n=1}^N \beta_n \frac{1}{(2\pi)|\boldsymbol{\widehat{\Sigma}}_n|^{1/2}} \exp\left(-\frac{1}{2} \l(\boldsymbol{x}-\boldsymbol{\hat{\mu}}_n
\r) \boldsymbol{\hat{\Sigma}}_n^{-1} (\boldsymbol{x}-\boldsymbol{\hat{\mu}}_n)\right).
\ENS
Here, for $n \in N$,  $\boldsymbol{\hat{\mu}}_n = [\hat{\mu}_n^{(1)}, \hat{\mu}_n^{(2)}]$,
$\boldsymbol{\widehat{\Sigma}}_n$ is the covariance matrix, where
$\boldsymbol{\widehat{\Sigma}}^{(\ell, k)}_n = \hat{\sigma}^{(\ell)}_n \hat{\sigma}^{(k)}_n \hat{\rho}^{(\ell, k)}_n$,
$\ell, k \in  \{1, 2\}$, with  $\hat{\rho}^{(1,2)}_{n} = \hat{\rho}^{(2,1)}_{n} = \hat{\rho}_n \in [-1, 1]$.
The parameters to be learned are: $\{\beta_n, \hat{\mu}_n^{(1)}, \hat{\mu}_n^{(2)},  \hat{\rho}_n\}$, $n = 1, \ldots N$.
The real and imaginary parts of the Fourier transform of $\widehat{g}(\boldsymbol{x};\theta)$ are given by
\begin{align*}
    \text{Re}_{\widehat{G}}(\boldsymbol{\eta}) \!= \!\!\sum_{n=1}^N \beta_n \cos(\boldsymbol{\eta}'\boldsymbol{\hat{\mu}}_n)\exp\bigg(\frac{-\boldsymbol{\eta}'\boldsymbol{\widehat{\Sigma}}_n\boldsymbol{\eta}}{2}\bigg),
    ~
    \text{Im}_{\widehat{G}} (\boldsymbol{\eta}) \!= \!\!\sum_{n=1}^N \beta_n \sin(\boldsymbol{\eta}'\boldsymbol{\hat{\mu}}_n)\exp\bigg(\frac{-\boldsymbol{\eta}'\boldsymbol{\widehat{\Sigma}}_n\boldsymbol{\eta}}{2}\bigg).
\end{align*}
We conduct numerical experiments using the parameters listed in Table~\ref{tab:MJD2_parameters}, with estimation errors detailed in Table~\ref{tab:MJD2_error_g}.
As evident from Table~\ref{tab:MJD2_error_g}, FourNet demonstrates impressive $L_2$ estimation accuracy.
{\siscb{Here, \( 10^3 \) partition points per dimension are used, totaling \( 10^6 \) data points for training ($P = 10^6$, as shown in Table~\ref{tab:hyper}).
Comparing the $L_2$ estimation errors for the two-dimensional and one-dimensional cases (Table~\ref{tab:MJD2_error_g} and Table~\ref{tab:MJD_error_g}), and in view of the multi-dimensional error bound~\eqref{eq:g_g_d}, it appears that FourNet is robust, accurate, and reliable even in higher dimensions.}}

\subsection{Option pricing}
We now turn our attention to the application of estimated transition densities produced by FourNet utilized for European and Bermudan option pricing. Recall that  $0\le t < t + \Delta t \le T$, where $t$ and $\Delta t$ are fixed.
Typically, in option pricing, we need to approximate a generic convolution integral of the form
\begin{linenomath}
\postdisplaypenalty=0
\begin{align}
\label{eq:nn_v1}
v(x, t) &=   e^{-r\Delta t} \int_{\Rbb} v((x'-c)/a, t + \Delta t) g(x-x'; \Delta t)~ dx'
\nonumber
\\
&\approx
e^{-r\Delta t} \int_{x_{\min}}^{x_{\max}} v((x'-c)/a, t + \Delta t) \widehat{g}(x-x';\widehat{\theta}, \Delta t)~ dx.
\end{align}
\end{linenomath}
Here, $v(\cdot, t + \Delta t)$ is the time-$(t + \Delta t)$ condition; $\widehat{g}(x;\widehat{\theta}, \Delta t)$ is the estimated transition density obtained through FourNet. As we pointed out in Remark~\ref{rm:linear}, $\widehat{g}(x;\widehat{\theta}, \Delta t)$ reflects a linear transformation applied to the original density. In light of this, the time-$(t + \Delta t)$ terminal condition must be adjusted correspondingly in the convolution integral, as depicted in \eqref{eq:nn_v1}, through $(x'-c)/a$. As noted earlier, this integral is evaluated using  adaptive Gauss quadrature rule (based on QUADPACK library in Fortran 77 library, \texttt{quad} function in Python).
The range $[x_{\min}, x_{\max}]$ for numerical integration will be provided for each test case subsequently.
\\
\begin{minipage}[t]{0.475\textwidth}
\strut\vspace*{-\baselineskip}\newline
\flushleft
\begin{tabular}{|p{0.1\textwidth}>{\raggedleft}p{0.2\textwidth}>{\raggedleft}p{0.2\textwidth}>{\raggedleft\arraybackslash}p{0.2\textwidth}|}
\hline
Strike & Ref.\  \cite{merton1976option}& FourNet   & Rel.\ \\
($E$) &     & \texttt{quad}   & error \\
\hline
96& 14.83787& 14.83790 & $2 \times 10^{-6}$  \\
\hline
98& 13.43922&  13.43925 & $3 \times 10^{-6}$ \\
\hline
100& 12.10782& 12.10785 & $3 \times 10^{-6}$ \\
\hline
102& 10.84925& 10.84928& $3 \times 10^{-6}$ \\
\hline
104& 9.66805& 9.66808& $3 \times 10^{-6}$ \\
\hline
\end{tabular}
\captionof{table}{European call option prices under the Merton model
corresponding to Tables~\ref{tab:MJD_parameters} and \ref{tab:MJD_error_g};
$[x_{\min}, x_{\max}] = [-4, 1]$.}
\label{tab:Merton_error_v}
\vfill
\end{minipage}
~~
\begin{minipage}[t]{0.475\textwidth}
\strut\vspace*{-\baselineskip}\newline
\flushright
\begin{tabular}{|p{0.1\textwidth}>{\raggedleft}p{0.2\textwidth}>{\raggedleft}p{0.2\textwidth}>{\raggedleft\arraybackslash}p{0.2\textwidth}|}
\hline
Strike & Ref.\ \cite{Fang2008}  & FourNet   & Rel.\ \\
($E$) &   (COS)&  \texttt{quad}  & error
\\
\hline
96&  21.78472 & 21.78466 &$3 \times 10^{-6}$  \\
\hline
98& 20.77826  &  20.77819& $3 \times 10^{-6}$ \\
\hline
100& 19.81294 & 19.81288& $3 \times 10^{-6}$ \\
\hline
102& 18.88821  & 18.88815& $3 \times 10^{-6}$ \\
\hline
104& 18.00334 & 18.00328& $3 \times 10^{-6}$ \\
\hline
\end{tabular}
\captionof{table}{
European call option prices under CGMY dynamics
corresponding to data from Tables~\ref{tab:CGMY_parameters} and
\ref{tab:CGMY_error_g}; $[x_{\min}, x_{\max}] = [-4, 2]$.}
\label{tab:CGMY_error_v}
\vfill
\end{minipage}
\\
\begin{minipage}[t]{0.475\textwidth}
\strut\vspace*{-\baselineskip}\newline
\flushleft
\begin{tabular}{|p{0.1\textwidth}>{\raggedleft}p{0.2\textwidth}>{\raggedleft}p{0.2\textwidth}>{\raggedleft\arraybackslash}p{0.2\textwidth}|}
\hline
Strike & Ref.\ \cite{heston93}  & FourNet   & Rel.\ \\
($E$) &     &  \texttt{quad}  & error
\\
\hline
96& 57.35019&57.35014 &$1 \times 10^{-6}$  \\
\hline
98& 56.61132   &  56.61127& $1 \times 10^{-6}$ \\
\hline
100& 55.88119  & 55.88114& $1 \times 10^{-6}$ \\
\hline
102& 55.15980  & 55.15975& $1 \times 10^{-6}$ \\
\hline
104&54.44716  & 54.44711& $1 \times 10^{-6}$\\
\hline
\end{tabular}
\label{tab:hes_error_v}
\captionof{table}{
European call option prices under Heston dynamics
corresponding to data from Tables~\ref{tab:hes_parameters} and
\ref{tab:hes_error_g}; $[x_{\min}, x_{\max}] = [-4, 1]$.
}
\vfill
\end{minipage}
~~
\begin{minipage}[t]{0.475\textwidth}
\strut\vspace*{-\baselineskip}\newline
\flushright
\begin{tabular}{|p{0.1\textwidth}>{\raggedleft}p{0.2\textwidth}>{\raggedleft}p{0.2\textwidth}>{\raggedleft\arraybackslash}p{0.2\textwidth}|}
\hline
Strike & Ref.\ \cite{Fang2008}  & FourNet   & Rel.\ \\
($E$) &   (COS)&  \texttt{quad}  & error
\\
\hline
7& 4.27369&4.27373 &$1 \times 10^{-5}$  \\
\hline
8& 3.81734   &  3.81738& $1 \times 10^{-5}$ \\
\hline
9& 3.40704   & 3.40708& $1 \times 10^{-5}$ \\
\hline
10& 3.04018   & 3.04022& $1 \times 10^{-5}$ \\
\hline
11&2.71399  & 2.71403& $1 \times 10^{-5}$ \\
\hline
\end{tabular}
\label{tab:HQH_error_v}
\captionof{table}{European call option prices under Heston Queue-Hawkes dynamics;
corresponding to data from Tables~\ref{tab:HQH_parameters} and
\ref{tab:HQH_error_g}; $[x_{\min}, x_{\max}] = [-3, 1]$}.
\end{minipage}

\subsubsection{European options}
For European options, we set $t = 0$ and $\Delta t = T$, and $v(x', t + \Delta t) = v(x', T)$ as the payoff function.
The strike of the option is given by $E>0$.
In the context of exponential L\^{e}vy processes examined in this study, which include the Merton and CGMY dynamics, the European call option payoff function is defined as $v\big((x'-c)/a, T\big) \equiv \big(s_0 e^{(x'-c)/a} - E\big)^+$, where $E$ is strike price of the option.
For the Heston and Heston Queue-Hawkes models, the European call option payoff is $v\big((s'-c)/a, T\big) = \big(e^{(s'-c)/a} - E\big)^+$.

We provide numerically computed European option prices for the models discussed in the previous section. These are presented in Tables~\ref{tab:Merton_error_v} (Merton),  \ref{tab:CGMY_error_v} (CGMY), \ref{tab:hes_error_v} (Heston), and \ref{tab:HQH_error_v} (Heston Queue-Hawkes). Option prices are derived using the FourNet-estimated transition density \(\widehat{g}(s;\widehat{\theta}, \Delta t)\), combined with an adaptive Gauss quadrature rule (the \texttt{quad} function in Python) to evaluate the corresponding convolution integral. These results are displayed under the ``FourNet-\texttt{quad}'' column.

Benchmark prices are detailed under the ``Ref.'' column. For the Merton and Heston models, these benchmark prices are determined using the analytical solutions from \cite{merton1976option} and the method by \cite{heston93}, respectively. For CGMY, and Heston Queue-Hawkes models, reference European option prices are derived from our implementation of the Fourier Cosine (COS) method \cite{Fang2008}. The associated relative errors of these approximations are indicated in the ``Rel.\ error'' column. Clearly, the FourNet-\texttt{quad} method proves highly accurate, showcasing a negligible error (on the order of \(10^{-5}\)).

\subsubsection{Bermudan options}
We present a Bermudan put option written on the underlying following the Merton jump-diffusion model \cite{ForsythLabahn2017}.
Unlike European options which can only be exercised at maturity,
a Bermudan put option can be exercised at any fixed dates $t_m^-$, $t_m\in \mathcal{T}$, where
$\mathcal{T}  \equiv \{t_m\}_{m = 1}^M $ is a discrete set of pre-determined early exercise dates. We adopt the
convention that no early exercise at time $t_0$. In this example, the early exercise dates are annually apart, that is, $t_{m+1} - t_m = \delta t = 1$ (year), In addition, the underlying asset pays a fixed dividend amount $D$ at $t_m^-$.
Importantly, given that the transition density needed for pricing is time and spatially homogeneous, and with \(\delta t = 1\) year for all intervals, we can efficiently train the FFNN only once, and apply it across all intervals \([t_m, t_{m+1}]\).

Over each $[t_m, t_{m+1}]$, the pricing algorithm for a Bermudan put option consists of two steps.
In Step 1 (time-advancement), we need to approximate the convolution integral \eqref{eq:nn_v1}:
$v(x, t_m^+) = e^{-r\Delta t} \int_{\Rbb} v((x'-c)/a, t_{m+1}) \widehat{g}(x-x';\widehat{\theta}, \delta t)~ dx$,
for $x \in [x_{\min}, x_{\max}]$, where $ x_{\min} <0< x_{\max}$, $|x_{\min}|$ and $x_{\max}$ are sufficiently large.
In Step 2 (intervention), we impose the condition
\EQ
\label{eq:ber}
v(x, t_m) = \max\big(v\big(\ln(\max(e^x-D, e^{x_{\min}})),t_m^+\big), \max(E-e^x,0)\big).
\EN
Here, $E$ is the strike price, and the expression $\ln(\max(e^x-D, e^{x_{\min}}))$ in \eqref{eq:ber} ensures that the no-arbitrage condition
holds, i.e.\ the dividend paid can not be larger than the stock price at that time, taking into account the localized grid.

Adopting annual early exercise dates, where \(\delta t = 1\) year, the transition density \(\widehat{g}(\cdot; T= 1)\) as obtained from FourNet (as detailed in Table~\ref{tab:MJD_error_g} and based on parameters from Table~\ref{tab:MJD_parameters}) is used in Step~1 (time-advancement) above. 
These parameters and those pertaining to the Bermudan put are given in Table~\ref{tab:berm_parameters}.

Letting $\{x_q\}_{q = 0}^Q$ be a partition of $[x_{\min}, x_{\max}]$, we denote by $v_q^m$ a numerical approximation to the exact value $v(x_q, t_m)$,
where $t_m \in \mathcal{T} \cup\{t_0\}$. Intermediate value $\{v_q^{m+}\}$, $q = 0, \ldots, Q$, is computed by evaluating
the convolution integral in Step~1 via an adaptive Gauss quadrature rule, specifically the \texttt{quad} function in Python.
The time $t_{m+1}$-condition $v(\cdot, t_{m+1})$ is given by a linear combination of discrete solutions $\{v_q^{m+1}\}$.
For condition \eqref{eq:ber}, linear interpolation is then used on $\{v_q^{m+}\}$  to determine the option value $\{v_q^m\}$.
Convergence results for the FourNet-quad approach are displayed in Table~\ref{tab:Merton_error_v_ber}, showcasing evident agreement with an accurate benchmark option price taken from \cite{ForsythLabahn2017}[Table~5] (finest grid).
{\siscblue{Here, to estimate the convergence rate of FourNet-quad, we calculate the ``change'' as the difference in values from coarser to finer partitions (i.e.\ transitioning from smaller to larger values of $Q$) and the ``ratio'' as the quotient of these changes between successive partition.}}
\\
\begin{minipage}[t]{0.45\textwidth}
\strut\vspace*{-\baselineskip}\newline
\flushleft
\begin{tabular}{|p{0.7\textwidth}>{\raggedleft\arraybackslash}p{0.15\textwidth}|}
\hline
Parameters              & Values \\
\hline
$S_0$ (initial asset price)                          &100\\
$r$ (risk free rate)                        & 0.05  \\
$\sigma$ (volatility)                   & 0.15  \\
$\lambda$ (jump intensity)                  & 0.1  \\
$\tilde{\mu}$ (mean of jump size)               & -1.08 \\
$\tilde{\sigma}$ (std of jump size)           & 0.4  \\
\hline
$T$ (maturity in years)                           &10\\
$\delta t$ (frequency in years)                   & 1 \\
$E$ (strike)                          &100\\
$D$ (dividend)                   & 1  \\
\hline
\end{tabular}
\captionof{table}{Parameters for the Bermudan put option;
values taken from \cite{ForsythLabahn2017}[Table~4].}
\label{tab:berm_parameters}
\end{minipage}
~~~
\begin{minipage}[t]{0.45\textwidth}
\strut\vspace*{-\baselineskip}\newline
\flushleft
\vspace*{-0.35cm}
\begin{tabular}{|>{\raggedright}p{0.1\textwidth}>{\raggedleft}p{0.25\textwidth}>{\raggedleft}p{0.25\textwidth}
>{\raggedleft\arraybackslash}p{0.15\textwidth}|}
\hline
$Q$ & FourNet-quad  & \siscblue{change} & ratio\\
\hline
200& 24.8323& &\\
\hline
400& 24.7903& $\siscblue{4.2 \times 10^{-2}}$ &\\
\hline
800& 24.7838 &$\siscblue{6.5\times 10^{-3}}$&6.7\\
\hline
1600& 24.7812&$\siscblue{2.6 \times 10^{-3}}$ &2.5\\
\hline
3200& 24.7806&$\siscblue{6.0\times 10^{-4}}$  & 4.3\\
\hline
\end{tabular}
\captionof{table}{Bermudan put option; parameters from Table~\ref{tab:berm_parameters}; benchmark price:  24.7807, taken from \cite{ForsythLabahn2017}[Table~5, finest grid]; $x_{\min} = \ln(S_0)- 10$, $x_{\max} = \ln(S_0) + 10$.}
\label{tab:Merton_error_v_ber}
\vfill
\end{minipage}

\subsection{Robustness tests}
This section evaluates FourNet's robustness against the COS method \cite{Fang2008} in challenging scenarios, particularly when the transition density approaches a Dirac's delta function as \( T \to 0 \). Fourier-based methods often struggle in such cases, requiring a very large number of terms to accurately estimate the transition density. To further assess robustness, we introduce an asymmetric heavy-tailed distribution. We test with a very short maturity \( T = 0.001 \) years, or about 0.26 business days, using the Kou jump-diffusion model \cite{kou2002jump}, known for effectively modeling the leptokurtic nature of market returns.
\\
\begin{minipage}[t]{0.475\textwidth}
\strut\vspace*{-\baselineskip}\newline
\flushleft
\begin{tabular}{|p{0.7\textwidth}>{\raggedleft\arraybackslash}p{0.15\textwidth}|}
\hline
Parameters              & Values \\

\hline
$T$ (maturity in years)                           &0.001\\
$S_0$ (initial asset price)                           &100\\
$q_1$ (jump-up probability)                          &0.3445\\
$r$ (risk free rate)                        & 0.05  \\
$\sigma$ (volatility)                   & 0.15  \\
$\lambda$ (jump intensity)                  & 0.1  \\
$\xi_1$ (jump-up param.)               & 3.0465 \\
$\xi_2$ (jump-down param.)           & 3.0775 \\
\hline
$E$ (strike price)                          &100\\
\hline
\end{tabular}
\captionof{table}{
Parameters for the Kou jump diffusion dynamics;
values are taken from \cite{ForsythLabahn2017}[Table 1].} \label{tab:kou1_parameters}
\end{minipage}
~~
\begin{minipage}[t]{0.45\textwidth}
\strut\vspace*{-\baselineskip}\newline
\flushright
\begin{tabular}{|p{0.45\textwidth}>{\raggedleft\arraybackslash}p{0.35\textwidth}|}
\hline
$N$ ($\#$ of neurons)  & 45
\\
\hline
$L_1\l(\text{Re}_{G}, \text{Re}_{\widehat{G}}\r)$                           &$1.06\times10^{-03}$
\\
\rowcolor{gray!40}
$L_2\l(\text{Re}_{G}, \text{Re}_{\widehat{G}}\r)$                           &$1.86\times10^{-08}$
\\
$\text{MPE}\l(\text{Re}_{G}, \text{Re}_{\widehat{G}}\r)$  & $7.00\times10^{-05}$
\\
\hline
$L_1\l(\text{Im}_{G}, \text{Im}_{\widehat{G}}\r)$ & $1.10\times10^{-03}$
\\
\rowcolor{gray!40}
$L_2\l(\text{Im}_{G}, \text{Im}_{\widehat{G}}\r)$ & $2.11\times10^{-08}$
\\
$\text{MPE}\l(\text{Im}_{G}, \text{Im}_{\widehat{G}}\r)$ & $3.16\times10^{-05}$
\\
\hline
\end{tabular}
\captionof{table}{
Estimation errors for the Kou model;
parameters from  Table~\ref{tab:kou1_parameters};
linear transform in Remark~\ref{rm:linear} is employed with $(a,c)=(20, 0)$.
}
\label{tab:kou1_error_g}
\vfill
\end{minipage}

\vspace*{+0.25cm}
For this model, the characteristic function of the random variable $X_t = \ln\l(S_t/S_0\r)$ is $G_X(\eta) = \exp(t \psi(\eta))$, where
$\psi(\eta) = i\left(\mu-\frac{\sigma^2}{2}\right) \eta-\frac{\sigma^2 \eta^2}{2} +\lambda\left(\frac{q_1}{1-i \eta \xi_{1}}+\frac{q_2}{1+i \eta \xi_{2}}-1\right)$, with $q_1 \in (0, 1)$ and $q_1 + q_2 = 1$, $\xi_1 > 1$ and $\xi_2 > 0$.

The parameters for this experiment are detailed in Table~\ref{tab:kou1_parameters}, with the  parameters for the linear transformation set to $(a,c)=(20, 0.0)$. FourNet's training results, shown in Table~\ref{tab:kou1_error_g}, reveal an $L_2$-estimation error around $10^{-8}$, demonstrating significant accuracy in this challenging scenario.

In comparison, the COS method, using 800 and 1200 terms (COS-800 and COS-1200) as implemented per \cite{Fang2008}, exhibited significant oscillations and losses of non-negativity in estimated transition densities, especially near \(x = 0\), with these issues being particularly pronounced in the right tail (Figure~\ref{fig:kou_p}). Similar issues were noted in the left tail but are not shown. Figure~\ref{fig:kou_p_call} further illustrates how these oscillations and non-negativity issues in the COS method can lead to highly fluctuating and sometimes negative European call option prices, violating the no-arbitrage principle. In contrast, FourNet displayed minimal non-negativity loss, demonstrating its robustness and precision for financial applications. Notably, compared to European option prices calculated using an analytical formula from \cite{kou2002jump}, FourNet achieved maximum relative errors in option prices of about $10^{-3}$, showcasing superior accuracy.

\begin{figure}[htb!]
\vspace*{-0.5cm}
	\centering
		\subfigure[Estimated transition density]{
		\includegraphics[width = 0.44\textwidth] {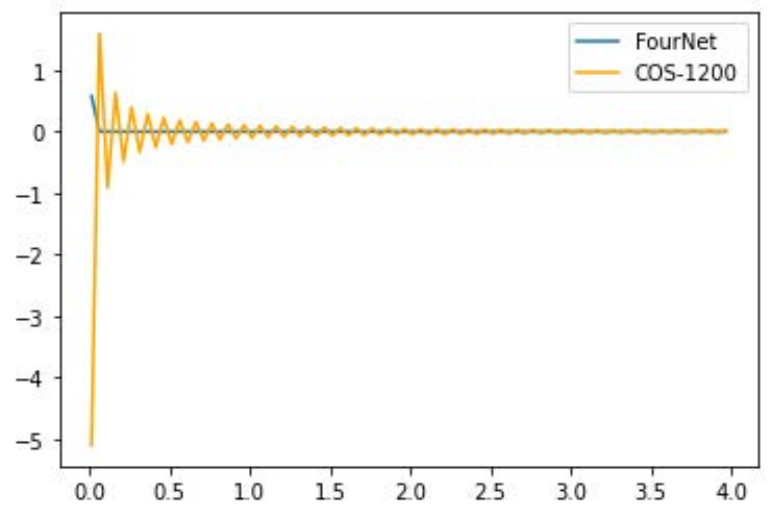}
        \label{fig:kou_p}
        }
        \subfigure[Option prices]{
		\includegraphics[width = 0.46\textwidth] {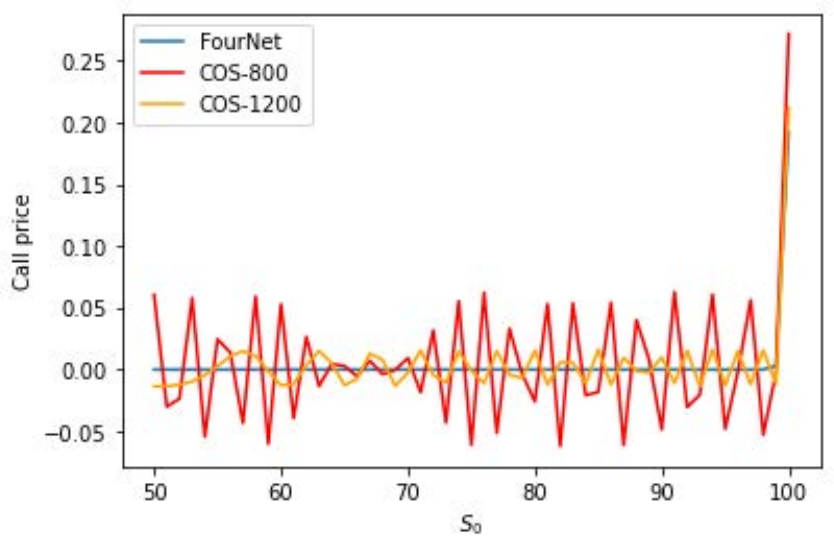}
        \label{fig:kou_p_call}
        }
	\caption{Comparison between FourNet and COS-800/COS-1200,  corresponding to parameters/data from Table~\ref{tab:kou1_parameters} and Table~\ref{tab:kou1_error_g}. }
\label{fig:kou_p_two}
\end{figure}

\section{Conclusion and future work}
\label{sc:cc}
This paper has introduced and rigorously analyzed FourNet, a novel single-layer FFNN developed to approximate transition densities with known closed-form Fourier transforms. Leveraging the unique Gaussian activation function, FourNet not only facilitates exact Fourier and inverse Fourier operations, which is crucial for training,  but also draws parallels with the Gaussian mixture model, demonstrating its power in approximating sufficiently well a vast array of transition density functions. The hybrid loss function, integrating MSE with MAE regularization, coupled with a strategic sampling approach, has significantly enhanced the training process.

Through a comprehensive mathematical analysis, we demonstrate FourNet's capability to approximate transition densities in the $L_2$-sense arbitrarily well with a finite number of neurons. 
We derive practical bounds for the $L_2$ estimation error and the potential (pointwise) loss of nonnegativity in the estimated densities {\siscb{for the general case of $d$-dimensions ($d \ge 1$)}}. We derive practical bounds for the $L_2$ estimation error and the potential (pointwise) loss of nonnegativity in the estimated densities, underscoring the robustness and applicability of our methodology in practical settings. We illustrate FourNet's accuracy and versatility through a broad range of models in quantitative finance, including (multi-dimensional) exponential L\'{e}vy processes and the Heston stochastic volatility models-even those augmented with the self-exciting Queue-Hawkes jump process. European and Bermudan option prices computed using estimated transition densities obtained through FourNet exhibit impressive accuracy.

In future work, we aim to extend FourNet to tackle more complex stochastic control problems, potentially involving higher dimensionality and model non-homogeneity. This expansion is expected to broaden FourNet's applicability and enhance its utility in sophisticated financial modeling. We plan to explore various approaches to improve its performance in high-dimensional settings, assessing a range of enhancements to optimize its architecture and training processes. In addition, FourNet's simplicity and ease of implementation position it well for realistic models previously deemed challenging within existing frameworks. One particular area of interest includes investigating the impact of self-exciting jumps on optimal investment decisions in Defined Contribution superannuation--a topic of heightened relevance in a climate marked by rising inflation and economic volatility.

\small

\section*{Appendices}
\appendix

\section{Constructing non-uniform partitions with multiple peaks}
\label{app:non-uni}
In Algorithm~\ref{alg:mapping_pro1}, we provide a detailed procedure for constructing non-uniform, yet fixed, partitions of the interval $[\eta_l, \eta_u]$, comprised of $M$ sub-intervals. These partitions feature denser points around a specifically chosen point, $\eta_c \in [\eta_l,\eta_u]$. The parameters $d_l$ and $d_u$ determine the point densities in the intervals $[\eta_l, \eta_c]$ and $[\eta_c, \eta_u]$, respectively, represented as $\frac{1}{d_l}$ and $\frac{1}{d_u}$.

\begin{algorithm}[h]{\emph{PartitionOne}}($\eta_l, \eta_u, \eta_c, M, m, d_l, d_u$)
\caption{
\label{alg:mapping_pro1}
Algorithm for constructing a non-uniform partition of an interval $[\eta_l, \eta_u]$ into $M$ sub-intervals,
having a single concentration point, $\eta_c$, which is the $m$-partition point, $m \in\{ 0, \ldots, M\}$, is fixed. }
\begin{algorithmic}[1]
\STATE compute $\ds \alpha_l = \sinh^{-1}\Big{(}\frac{\eta_l-\eta_c}{d_l}\Big{)}$ and
$\ds \alpha_u = \sinh^{-1}\Big{(}\frac{\eta_u-\eta_c}{d_u}\Big{)}$;

\STATE
compute $\eta_0 = \eta_l; \ds \eta_j = \eta_c + d_l\sinh(\alpha_l(1-k_j))$,
where $\ds k_j = \frac{j}{m}$, $~j = 1, \ldots, m$;

\STATE
compute $\ds \eta_j = \eta_c + d_u\sinh(\alpha_u k_j)$,
where $\ds k_j = \frac{j}{M-m}$, $~j = 1, \ldots, (M-m)$;

\STATE
return $\ds Q \equiv \{\eta_j\}_{j = 0}^m  \cup \{\eta_j\}_{j = 1}^{M-m}$;
\end{algorithmic}
\end{algorithm}

\begin{algorithm}[h]{\emph{PartitionMulti}}($\eta_{\min}, \eta_{\max}, \{\eta_j\}_{j=1}^J, \{P_j\}_{j=1}^v, \{q_j\}_{j=1}^J, \{\eta_l^j\}_{j=1}^J,
\{\eta_u^j\}_{j=1}^J $)
\caption{
\label{alg:mapping_pron}
Algorithm for constructing a non-uniform partition of an interval
with multiple concentration points.}
\begin{algorithmic}[1]
\STATE
$\ds Q_1 \leftarrow PartitionOne\Big{(}\eta_{\min}, \frac{\eta_{1} + \eta_{2}}{2}, \eta_1, P_1, q_1, \eta_l^1, \eta_u^1\Big{)}$;

\STATE
$\ds Q_j \leftarrow PartitionOne(\frac{\eta_{j-1} + \eta_{j}}{2}, \frac{\eta_{j} + \eta_{j+1}}{2}, \eta_j, P_j, q_j, \eta_l^j, \eta_u^j)$,
$j = 2, \ldots, J-1$;

\STATE
$\ds Q_J \leftarrow PartitionOne(\frac{\eta_{J-1} + \eta_{J}}{2}, \eta_{\max},
\eta_J, P_J, q_J, \eta_l^J, \eta_u^J)$;

\STATE
return $\ds Q \equiv \cup_{j=1}^J Q_j$;
\end{algorithmic}
\end{algorithm}
We use Algorithm~\ref{alg:mapping_pro1} in Algorithm~\ref{alg:mapping_pron} to generate a non-uniform partition having $P$ sub-intervals for the region $[\eta_{\min}, \eta_{\max}] \equiv [-\eta', \eta']$ with concentration points $\eta_j$, $j = 1, \ldots, J$, satisfying $\ds \eta_{\min} \leq \eta_1< \eta_2 < \ldots < \eta_J \leq \eta_{\max}$.
Here, $P_j$ is the number of sub-intervals for the $j$-th sub-region containing $\eta_j$, $j = 1, \ldots, J$,
with $\sum_{j=1}^J P_j = P$;
$q_j$ is the local
index of the gridpoint in the $j$-th sub-region that is equal
to $\eta_j$; $\eta_l^j$ and $\eta_u^j$ are the upper and lower
density parameters, respectively, associated with the $j$-th sub-region containing $\eta_j$.
\end{document}